\documentclass{rQUF2e}

\usepackage{epstopdf}
\usepackage{subfigure}
\usepackage{threeparttable} 
\usepackage{booktabs} 
\usepackage{comment} 
\usepackage{color} 
\usepackage{rotating} 
\usepackage{lscape} 

\newcommand{\ep}{\varepsilon}

\usepackage{tikz}
\usetikzlibrary{calc}
\usetikzlibrary{shapes.geometric}
\usetikzlibrary{arrows.meta}

\newcommand{\bx}{\bm{x}}
\newcommand{\bX}{\bm{X}}
\newcommand{\by}{\bm{y}}
\newcommand{\bY}{\bm{Y}}

\newcommand{\bZ}{\bm{Z}}

\newcommand{\bw}{\bm{w}}
\newcommand{\bW}{\bm{W}}
\newcommand{\bu}{\bm{u}}

\newcommand{\bbeta}{\bm{\beta}}
\newcommand{\bmu}{\bm{\mu}}

\newcommand{\bone}{\bm{1}}
\newcommand{\bzero}{\bm{0}}
\newcommand{\bep}{\bm{\varepsilon}}

\newcommand{\bI}{\bm{I}}

\newcommand{\R}{\mathbb{R}}

\newcommand{\E}{\mathbb{E}}

\newcommand{\Var}{\text{Var}}

\newcommand{\Prob}{\mathbb{P}}

\newcommand{\dd}{\text{d}}

\newcommand{\sqc}[2]{#1_{1},#1_{2},\dots,#1_{#2}}

\newcommand{\AC}[1]{\text{AC}_{#1}}
\newcommand{\estAC}[2]{\widehat{\text{AC}}^{\text{\scriptsize #1}}_{#2}}

\newcommand{\tra}[1]{#1^{\text{\scriptsize T}}}

\newcommand{\iidsim}{\stackrel{i.i.d.}{\sim}}
\newcommand{\darrow}{\stackrel{d}{\longrightarrow}}

\theoremstyle{plain}
\newtheorem{theorem}{Theorem}[section]

\theoremstyle{definition}

\newtheorem{example}{Example}

\newtheorem{assumption}{Assumption}

\theoremstyle{remark}
\newtheorem{remark}{Remark}

\newcommand*{\changed}[1]{}

\begin{document}


\title{Estimation of risk contributions with MCMC}

\author{TAKAAKI KOIKE$^{\ast}$$\dag$\thanks{$^\ast$Corresponding author.
Email: tkoike@uwaterloo.ca} and MIHOKO MINAMI${\ddag}$\\
\affil{$\dag$Department of Statistics and Actuarial Science, University of Waterloo, 200 University Avenue West, Waterloo, ON, N2L 3G1, Canada.\\
$\ddag$Department of Mathematics, Keio University, Hiyoshi, Kohoku-ku, Yokohama, Kanagawa, 3-14-1, Japan.}} 

\maketitle

\begin{abstract}
Determining risk contributions of unit exposures to portfolio-wide economic capital is an important task in financial risk management.
Computing risk contributions involves difficulties caused by rare-event simulations.
In this study, we address the problem of estimating risk contributions when the total risk is measured by value-at-risk (VaR).
Our proposed estimator of VaR contributions is based on the Metropolis-Hasting (MH) algorithm, which is one of the most prevalent Markov chain Monte Carlo (MCMC) methods.
Unlike existing estimators, our MH-based estimator consists of samples from conditional loss distribution given a rare event of interest.
This feature enhances sample efficiency compared with the crude Monte Carlo method.
Moreover, our method has the consistency and asymptotic normality, and is widely applicable to various risk models having joint loss density.
Our numerical experiments based on simulation and real-world data demonstrate that in various risk models, even those having high-dimensional ($\approx 500$) inhomogeneous margins, our MH estimator has smaller bias and mean squared error compared with existing estimators. 
\\ \\
\end{abstract}

\begin{keywords}
Value-at-risk; Risk allocation; Risk contributions; VaR contributions; Copulas; Markov chain Monte Carlo; Metropolis-Hastings algorithm\\
\end{keywords}

\begin{classcode}C58, C63\end{classcode}

\section{Introduction}\label{introduction}
In most financial institutions, the risk of their portfolios is measured by economic capital.
Capital allocation is an important risk analysis, where the economic capital is decomposed into a sum of \emph{risk contributions} of unit exposures; see, for example, \citet{dev2004economic}.
The \emph{Euler principle}, proposed in \citet{tasche1999risk}, is one of the most well-known rules of risk allocation.
It is economically justified, for example, in \citet{denault2001coherent} and \citet{tasche1999risk,tasche2008capital}\changed{, and the resulting allocated capital is also known as the \emph{Aumann-Shapley value} \citep{aumann2015values} for risk capital allocation problems; see, for example, \citet{boonen2012generalization}, \citet{kalkbrener2005axiomatic}, and \citet{myers2001capital}.}

On the other hand, calculating risk contributions poses theoretical and numerical difficulties, especially when the portfolio-wide risk is measured by {\it value-at-risk (VaR)}.
Although a simple formula of VaR contributions is derived by \citet{tasche2001conditional}, it can rarely be calculated analytically without a few exceptions, for example, in \citet{tasche2004capital}.
As is seen in \citet{fan2012decomposition} and \citet{yamai2002comparative}, the \emph{crude Monte Carlo (MC)} method is the simplest method of computing risk contributions.
However, the MC estimator suffers from unignorable bias caused by sample inefficiency and by inevitable numerical modification; see, for instance, \citet{yamai2002comparative}. To overcome such difficulties, several methods have been proposed in the literature.
For instance, \citet{hallerbach2003decomposing} and \citet{tasche2004approximations} derived approximation formulas by regarding VaR contributions as the best predictor of individual losses given total loss.
In this paper, we call this estimator the {\it generalized regression (GR) estimator}.
\citet{glasserman2005measuring} developed {\it importance sampling (IS) estimators} with their main focus on credit portfolios.
\changed{The IS method generates samples from the so-called \emph{instrumental distribution} and then adjusts them so that the estimator is consistent.
This process of adjustment typically increases the variance of the estimator as a trade-off for enhancing effective sample size.
Finding appropriate instrumental distributions to reduce the variance often relies on knowledge of the distribution on rare events of interest. 
}
Finally, \citet{tasche2009capital} proposed the {\it Nadaraya-Watson (NW) estimator}, which is based on the kernel estimation method.
Despite its ease of calculation, it still requires importance sampling to achieve an efficient estimation.

In this paper, we propose a new method of estimating VaR contributions that utilizes the {\it Markov chain Monte Carlo (MCMC)}, especially the \emph{Metropolis-Hastings (MH) algorithm} \citep{metropolis1953equation,hastings1970monte}.
Our MH method requires joint loss density which can be evaluated at each point.
This is often the case when losses are modelled separately by marginal distributions and a copula; see \cite{yoshiba2013risk} for various examples.
\changed{For such loss models, the IS method is not straightforwardly applicable, since in general there is no guidance on the appropriate choice of instrumental distribution}. 
To the best of our knowledge, no stable estimator of VaR contributions is known for general risk models.
We study the consistency and asymptotic normality of our MH estimator, and provide practical guidelines for the efficient application of the MCMC method to the problem of computing VaR contributions.
The proposed method is then carried out for various risk models based on simulations and real-world data.
In numerical experiments, we compare the performance of the MH estimator with other existing estimators.

The foremost difference between our MH method and the crude MC is that in the former, samples are generated directly from the joint loss distribution given a rare event of interest.
In contrast, the MC method generates samples from the unconditional loss distribution, which makes it inevitable to waste a large portion of samples;
\changed{see Figure~\ref{Fig:Difference:MC:MCMC} for the illustration of this difference of MC and MCMC (MH).
The left figure of the MH samples also describes an underlying idea for simulating the conditional distribution of interest, which is to update samples sequentially so that they lie in the area of the rare event.
This feature of the MH method significantly improves sample efficiency without the need for any computational modification as is done in the MC method, which inevitably causes a bias.
As a consequence, a small bias and variance can be expected for the MH estimator compared with existing estimators.}

\begin{figure}[t]
\includegraphics[scale=0.4]{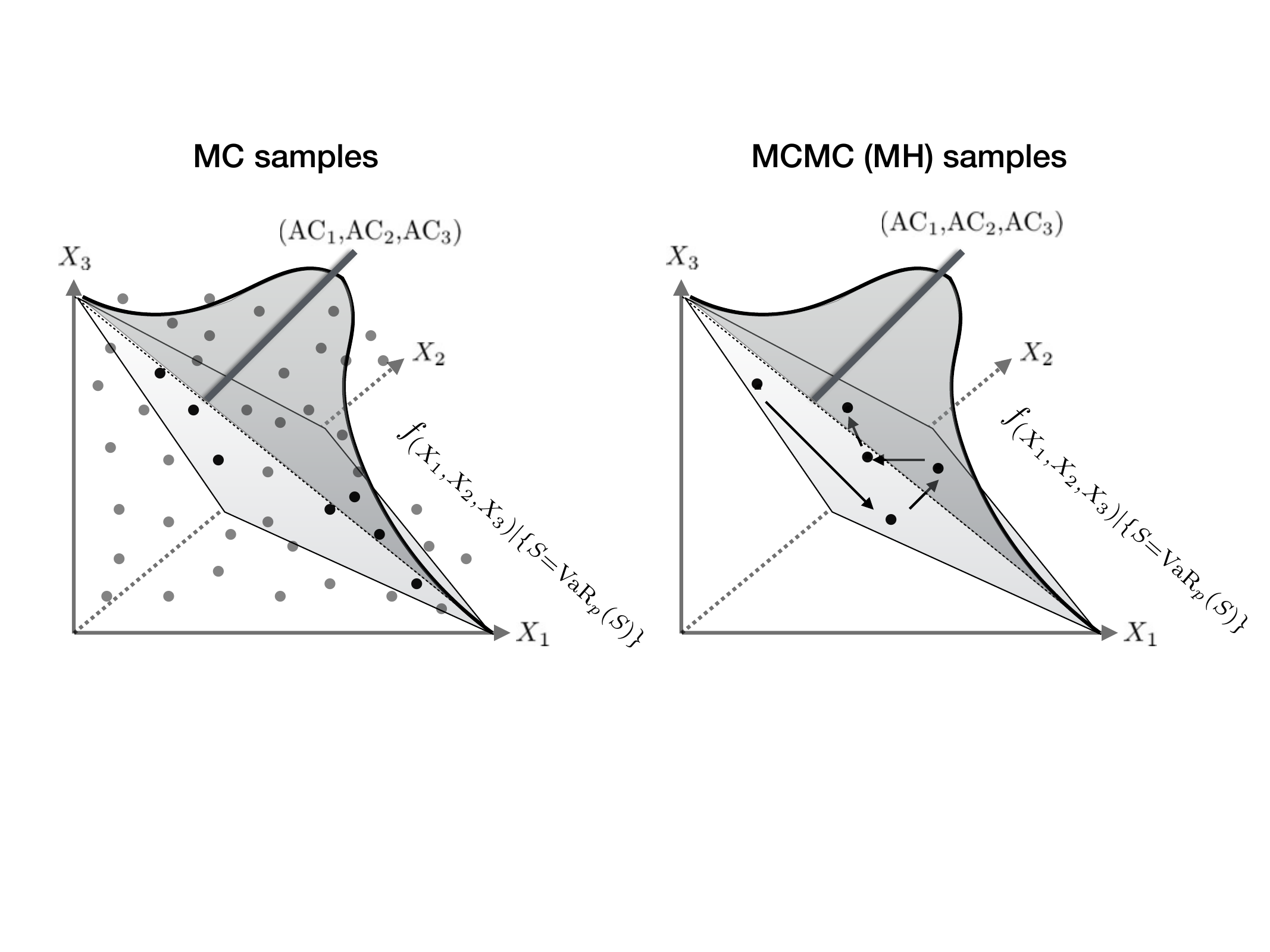}
\centering \vspace{-30mm}
\caption{
	The difference of the Monte Carlo (MC, left) and Markov chain Monte Carlo (MCMC, right) on estimating the VaR contributions (AC$_1$,AC$_2$,AC$_3$)$=\E[(X_1,X_2,X_3)| S=\text{VaR}_p(S)]$ where $X_j,j=1,2,3$ are loss random variables, $S=X_1+X_2+X_3$ is the total loss, and \text{VaR}$_p(S)$ is value-at-risk of $S$ with confidence level $p \in (0,1)$. 
	In the MC method, samples are generated from the unconditional distribution of $(X_1,X_2,X_3)$; a few samples close enough to the plane $\{(x_1,x_2,x_3)|x_1+x_2 + x_3 = \text{VaR}_{p}(S)\}$ are only used to estimate allocated capital.
	On the other hand, the MH method generates samples directly from the joint loss distribution given a rare event of interest, which is denoted as $f_{(X_1,X_2,X_3)|\{S=\text{VaR}_{p}(S)\}}$.
	}
\label{Fig:Difference:MC:MCMC}
\end{figure}

This paper is organized as follows.
Section \ref{capital allocation problem} introduces the mathematical setting of the capital allocation problem and explains challenges on estimating VaR contributions with the existing estimators.
Section \ref{sec:MCMC estimators} provides a brief introduction to the MCMC method and various MH algorithms.
In section \ref{the proposed method}, we propose the MH estimator that combines the MH method with the estimation of VaR contributions.
Next, in section \ref{Simulation and empirical studies}, numerical studies are conducted based on simulation and real-world data.
We demonstrate that for various risk models with marginal- and dependence-inhomogeneity and/or high-dimensionality, the MH estimator has smaller bias and mean squared error (MSE) than those of existing estimators.
For applying our method to other risk models not presented in this paper, practical guidelines on the usage of the MH method are also provided.
Concluding remarks and discussions are given in section \ref{Concluding remarks}.
Based on the theory of MCMC, the consistency and asymptotic normality of our estimator are derived in appendix A

\section{Capital allocation problem}\label{capital allocation problem}

Throughout this paper, the aggregate loss
\begin{equation*}\label{portfolio loss model}
S=\sum_{j=1}^{d}X_{j},
\end{equation*}
is considered, where $d\geq 3 $ is the size of the portfolio, and $\sqc{X}{d}$ are random variables on an atomless probability space $(\Omega,\mathcal F,\Prob)$ that represent the losses incurred by exposures $j = 1,2,\dots,d$ within a fixed time period.
In this study, a positive value of a loss random variable represents a financial loss, and a negative loss is interpreted as a profit. 
Let $F_{\bX}$ be the joint distribution function (df) of $\bX=(\sqc{X}{d})$ with margins $\sqc{F}{d}$, and let $F_{S}$ be the df of the total loss $S$.
Assume that $F_{\bX}$ and $F_S$ have densities $f_{\bX}$ with marginal densities $\sqc{f}{d}$ and $f_{S}$, respectively.
According to Sklar's theorem \citep[see, for example,][]{nelsen2006introduction}, it holds that
\begin{equation*}\label{sklar's theorem}
F_{\bX}(\bx)=C(F_{1}(x_{1}),\dots,F_{d}(x_{d})),\quad \bx=(\sqc{x}{d})\in \R^{d},
\end{equation*}
where $C$ is called a \emph{copula} of $\bX$.
The density $f_{\bX}$ can be written by
\begin{equation}\label{sklar theorem density form}
f_{\bX}(\bx)=c(F_{1}(x_{1}),\dots,F_{d}(x_{d})) f_{1}(x_{1})\cdots f_{d}(x_{d}),\quad \bx \in \R^{d},
\end{equation}
where $c$ denotes the density of $C$.

As mentioned in section \ref{introduction}, computing risk contributions is an important task in risk management.
A standard procedure of determining risk contributions involves two steps.
The first step is to compute the economic capital $\varrho(S)$ for a risk measure $\varrho$.
Risk measures map a loss random variable to a capital buffer that is required to cover the loss over a predetermined period such as one year or two weeks.
One of the most popular risk measures is the VaR defined by $\text{VaR}_{p}(X)=\inf\{x \in \R : \Prob(X\leq x) \geq p\}$ where $p\in (0,1)$ is called the {\it confidence level}.
Another popular measure is the \emph{expected shortfall (ES)} defined by $\text{ES}_p(X)=\frac{1}{1-p}\int_{p}^1 \text{VaR}_q(X)\dd q$ for $\E[|X|]<\infty$.
The second step is to allocate the capital $\varrho(S)$ to \changed{individual} $d$-exposures.
Mathematically, capital allocation addresses the problem of determining the vector of allocated capitals $(\sqc{\text{AC}}{d})$ that satisfies the {\it full allocation property}
\begin{equation}\label{full allocation property}
\varrho(S)= \sum_{j=1}^{d}\text{AC}_{j}.
\end{equation}
The Euler principle derives such AC's by utilizing the well-known Euler rule for a function $\bu \mapsto \varrho(\tra{\bu}\bX)$:
\begin{equation}\label{euler rule}
\varrho(\tra{\bu} \bX)=\sum_{j=1}^{d}u_{j}\frac{
  \partial \varrho(\tra{\bu} \bX)
}{
  \partial u_{j}
},\quad \bu \in \Lambda,
\end{equation}
where $\Lambda \subset \R^{d} \backslash \{\bzero\}$ is an open set such that $\bone_d \in \Lambda$, and $\varrho$ is positive homogeneous, that is, $\varrho(\lambda X)=\lambda \varrho(X)$ for $\lambda>0$.
For
\begin{equation}\label{derivative forula of AC via euler principle}
\AC{j}^{\varrho}:=\left. \frac{
  \partial \varrho(\tra{\bu} \bX)
}{
  \partial u_{j}
}\right|_{\bu=\bone_d}, \quad j=1,2,\dots,d,
\end{equation}
the full allocation property ($\ref{full allocation property}$) holds for the vector (AC$_{1}^{\varrho},\dots,$AC$^{\varrho}_{d}$) by taking $\bu=\bone_d$ in equation ($\ref{euler rule}$).
\changed{In addition, the allocated capital derived according to the Euler principle is RORAC compatible, which means that the profitability of each asset is consistently signaled via its return and AC; see \citet{tasche2008capital}.}
Since VaR$_{p}$ is positive homogeneous, the Euler principle can be applied, and the corresponding risk contributions are given by
\begin{equation}\label{VaR contributions}
\AC{j}^{\text{\scriptsize VaR}_{p}}:=
\left. \frac{
  \partial \text{VaR}_{p}(\tra{\bu} \bX)
}{
  \partial u_{j}
}\right|_{\bu=\bone_d}=
\E[X_{j}|X_{1}+\cdots+X_{d}=\text{VaR}_{p}(S)].
\end{equation}
We call the vector $\text{AC}^{\text{\scriptsize VaR}_{p}}:=(\text{AC}_{1}^{\text{\scriptsize VaR}_{p}},\dots,\text{AC}_{d}^{\text{\scriptsize VaR}_{p}})$ the {\it VaR contributions}.
Since we mainly focus on this form of allocated capital in this study, we drop the superscript VaR$_{p}$ and write \eqref{VaR contributions} as AC$=$($\AC{1},\dots,\AC{d}$).
Note that other forms of allocated capitals are also possible; for example, when the risk measure is ES, the \emph{ES contribution} is derived as
\begin{equation}\label{ES contributions}
\AC{j}^{\text{\scriptsize ES}_{p}}:=
\left. \frac{
  \partial \text{ES}_{p}(\tra{\bu} \bX)
}{
  \partial u_{j}
}\right|_{\bu=\bone_d}=
\E[X_{j}|X_{1}+\cdots+X_{d} \geq \text{VaR}_{p}(S)]
\end{equation}
by positive homogeneity of ES; see \cite{tasche2001conditional} for derivations of the last equalities in \eqref{VaR contributions} and \eqref{ES contributions}.

Even when the joint density of the portfolio loss vector $f_{\bX}$ is given explicitly, the analytical computation of AC is not straightforward since it often requires the joint distribution of $(X_{j},S)$, which is in general difficult to derive.
\changed{A possible case when VaR/ES contributions can be explicitly derived is when $\bX$ is modelled through a \emph{multivariate regularly varying (MRV)} distribution. 
In this case, VaR$_{p}(\tra{\bu}\bX)$ and ES$_{p}(\tra{\bu}\bX)$ asymptotically have explicit formulas as $p \rightarrow 1$, and thus VaR and ES contributions can be explicitly derived via \eqref{derivative forula of AC via euler principle}; see \citet{kley2016risk}.
In spite of their potential appeals, we avert from these formulas since the MRV assumption requires marginal tail-homogeneity; moreover, these formulas depend on integrals with respect to spectral measures, which is beset with other difficulties.}
A possible numerical method to calculate VaR contributions is the crude {\it MC} method,
in which the {\it pseudo VaR contribution}
\begin{equation}\label{pseudo VaR contributions}
\text{AC}_{\delta}=\E[\bX\ |\ S \in [\text{VaR}_{p}(S)-\delta,\text{VaR}_{p}(S)+\delta]\hspace{1mm}],
\end{equation}
is computed for a sufficiently small bandwidth $\delta>0$.
Since the probability $\Prob(S \in [\text{VaR}_{p}(S)-\delta,\text{VaR}_{p}(S)+\delta])$ is positive, the right hand side of ($\ref{pseudo VaR contributions}$) can be written as 
\begin{equation*}\label{pseudo VaR contribution expectation form}
\text{AC}_{\delta}=\frac{
  \E[\bX 1_{[S \in A_{\delta}]}]
}{
  \Prob(S \in A_{\delta})
},\quad \text{where}\quad A_{\delta}=[\text{VaR}_{p}(S)-\delta,\text{VaR}_{p}(S)+\delta].
\end{equation*}
This expression allows one to construct the estimator of the pseudo VaR contributions given by
\begin{equation}\label{MC estimator}
\estAC{MC}{\delta,N}=\frac{
    \sum_{n=1}^{N} \bX^{(n)}  1_{[S^{(n)}\in A_{\delta}]} 
}{
  \sum_{n=1}^{N}1_{[S^{(n)}\in A_{\delta}]}
}=
\frac{1}{M_{\delta,N}}\sum_{n=1}^{N} \bX^{(n)}  1_{[S^{(n)}\in A_{\delta}]},
\end{equation}
where $N>0$ is the sample size; $\bX^{(1)},\dots,\bX^{(N)}$ are independent and identically distributed (i.i.d.) samples from $F_{\bX}$; $S^{(n)}:=X_{1}^{(n)}+\cdots+X_{d}^{(n)}$ are i.i.d. samples from $F_{S}$ for $n=1,\dots,N$; and $M_{\delta,N}:=\sum_{n=1}^{N}1_{[S^{(n)}\in A_{\delta}]}$ is the number of samples contained in $A_{\delta}$. 
We call $(\ref{MC estimator})$ the {\it MC estimator}.
By setting $\delta$ and $N$ as sufficiently small and large, respectively, one can expect that the MC estimator approximates the true VaR contributions.
Note that this method is available only when $\delta$ is positive, since $\Prob(S \in A_{0})=\Prob(S=\text{VaR}_{p}(S))=0$ by continuity of $F_{S}$.

As long as the i.i.d. samples from $F_{\bX}$ can be generated, one can estimate AC$_\delta$ by constructing the estimator \eqref{MC estimator}.
However, this estimator suffers from an inevitable bias \changed{caused by changing the condition $\{S=\text{VaR}_p(S)\}$ to $\{S\in A_{\delta}\}$.}
The bias of the MC estimator can be decomposed by 
\begin{equation*}
\estAC{MC}{\delta,N}-\text{AC}=b_{\delta}(N)+b(\delta),
\end{equation*}
where $b_{\delta}(N)=\estAC{MC}{\delta,N}-\text{AC}_{\delta}$ and $b(\delta)=\text{AC}_{\delta}-\text{AC}$.
$\delta$ should be taken as small as possible to reduce $b(\delta)$.
However, when $\delta$ is quite small, it is difficult to ensure a large enough sample size $M_{\delta,N}$ to keep the first term $b_{\delta}(N)$ small since $\E[M_{\delta,N}]=N \Prob(S \in A_{\delta})$, and $\Prob(S \in A_{\delta})$ is typically much less than $1-p$.

To overcome this problem, several estimators have been proposed in the literature.
First, \changed{\citet{glasserman2005importance} and \citet{glasserman2005measuring} proposed the \emph{IS estimator} in order to 
increase the number of samples that belong to $A_{\delta}$.
The IS estimator is defined by
\begin{equation}\label{IS estimator}
\estAC{IS}{\delta,N,\tilde f}=
\frac{1}{\tilde{M}_{\delta,N}}\sum_{n=1}^{N} \tilde{\bX}^{(n)} l(\tilde{\bX}^{(n)}) 1_{[\tilde{S}^{(n)}\in A_{\delta}]},\quad \text{where}\quad
l(\tilde{\bX}^{(n)})= \frac{f_{\bX}(\tilde{\bX}^{(n)})}{{\tilde f}(\tilde{\bX}^{(n)})};
\end{equation}
here $N>0$ is the sample size; $\tilde{\bX}^{(1)},\dots,\tilde{\bX}^{(N)}$ are i.i.d. samples from the distribution with density $\tilde f$, called the \emph{instrumental distribution}; ${\tilde S}^{(n)}:=\tilde{X}_{1}^{(n)}+\cdots+\tilde{X}_{d}^{(n)}$ for $n=1,\dots,N$; and $\tilde{M}_{\delta,N}:=\sum_{n=1}^{N}1_{[\tilde{S}^{(n)}\in A_{\delta}]}$.
The main challenge of this method is to choose an appropriate $\tilde f$, so that the event $\{\tilde{S}^{(n)}\in A_{\delta}\}$ occurs frequently while the likelihood ratio $l(\tilde{\bX}^{(n)})$ remains non-volatile enough; see \citet{huang2007computation,liu2015simulating,mausser2007economic}, and \citet{tasche2009capital} for successful applications of IS to estimate risk contributions achieved by utilizing special assumptions.
}
Second, the {\it NW kernel estimator} proposed in \citet{tasche2009capital} is defined by
\begin{equation}\label{nw estimator}
\estAC{NW}{\phi,h,N}=\frac{
    \sum_{n=1}^{N} \bX^{(n)}  \phi\left(\frac{S^{(n)}-\text{VaR}_{p}(S)}{\Delta}\right)
}{
  \sum_{n=1}^{N}\phi\left(\frac{S^{(n)}-\text{VaR}_{p}(S)}{\Delta}\right)
},
\end{equation}
where $\phi$ is the kernel density and $\Delta>0$ is the bandwidth.
Since this estimator can be interpreted as a smoothing modification of the MC estimator ($\ref{MC estimator}$) by kernel $\phi$, it shares the same bias trade-off explained above.
Furthermore, the bias and asymptotic standard deviation of the NW estimator \citep[see, for example,][]{hansen2009nonparametric} cannot be computed easily because they require an evaluation of the total loss density $f_{S}(s)$ at $s=\text{VaR}_{p}(S)$.
Finally, \citet{hallerbach2003decomposing} and \citet{tasche2004approximations} constructed estimators by assuming a regression model among the losses of the form:
\begin{equation*}\label{model assumption}
\bX=g_{\bbeta}(S)+\bep,
\end{equation*}
where $g_{\bbeta}(s): \R \rightarrow \R^{d}$ is a function parameterized by $\bbeta$, and $\bep$ is an error random vector such that $\E[\bep|S=\text{VaR}_{p}(S)]=\bzero$.
For an estimator $\hat{\bbeta}_{N}$ of $\bbeta$,
we call the following estimator the {\it GR estimator}:
\begin{equation}\label{gr estimator}
\estAC{GR}{g_{\bbeta},N}:=
g_{\hat \bbeta_{N}}(\text{VaR}_{p}(S)).
\end{equation}
Although this estimator is intuitive and can easily be computed, it is in general difficult to construct an appropriate model $g_{\bbeta}$ and estimator $\hat{\bbeta}_{N}$ of $\bbeta$, unless samples from $F_{\bX|S=\text{VaR}_{p}(S)}$ are available.
A notable exception is the case wherein $\bX$ follows an elliptical distribution.
In this case, the following result holds: 
\begin{equation}\label{VaR contributions elliptical case}
\E[\bX|S=\text{VaR}_{p}(S)]=\E[\bX]+\frac{\text{Cov}(\bX,S)}{\text{Var}(S)}(\text{VaR}_{p}(S)-\E[S]);
\end{equation}
see, for example, \citet[][]{mcneil2015quantitative}.
The true VaR contributions are then provided by setting $g_{\beta}(s)=\beta_{0}+\beta_{1}s$, where
\begin{equation}\label{true coefficients elliptical case}
\beta_{0}=\E[\bX]-\frac{\text{Cov}(\bX,S)}{\text{Var}(S)} \E[S] \quad \text{and}\quad \beta_{1}=\frac{\text{Cov}(\bX,S)}{\text{Var}(S)}.
\end{equation}
Since these coefficients are the minimizers of $\E[\bep^{2}]=\E[(\bX-\beta_{0}-\beta_{1}S)^{2}]$, the OLS estimators of $(\beta_0,\beta_1)$ are calculated based on the unconditional samples of $\bX$ and $S$ converges to the true parameters \eqref{true coefficients elliptical case} as $N \rightarrow \infty$.

\section{MCMC estimators}\label{sec:MCMC estimators}
As seen in section \ref{capital allocation problem}, the essential problem in estimating VaR contributions is that the conditional samples from $F_{\bX|S=\text{VaR}_{p}(S)}$ are unavailable.
To solve this problem, we introduce the MCMC method wherein a given distribution is simulated by constructing a Markov chain whose stationary distribution is the desired one.
By allowing Markovian-type dependence within the samples, the MCMC allows us to simulate a wide variety of distributions.
In this section, we briefly review MCMC, especially the Metropolis-Hastings algorithm as a major subclass of MCMC methods.

\subsection{A brief introduction to MCMC}\label{brief introduction to MCMC}
Let $E\subseteq \R^{d}$ be a set and $\mathcal E$ be a $\sigma$-algebra on $E$. 
A {\it Markov chain} is a sequence of $E$-valued random variables $(\bX^{(1)},\bX^{(2)},\dots)$ satisfying the Markov property;
\begin{equation*}\label{markov property}
\Prob(\bX^{(n+1)}\in A\ |\ \bX^{(k)}=\bx^{(k)},k\leq n)=\Prob(\bX^{(n+1)}\in A\ |\ \bX^{(n)}=\bx^{(n)}),
\end{equation*}
for all $n\geq 1$, $A \in \mathcal E$, and $\bx^{(1)},\dots,\bx^{(n)}\in E$.
A Markov chain is characterized by its \emph{stochastic kernel} $K: E \times \mathcal E \rightarrow \changed{[0,1]}$, given by $\bx \times A \mapsto K(\bx,A):=\Prob(\bX^{(n+1)}\in A|\bX^{(n)}=\bx)$.
If there exists a probability distribution $\pi$ such that $\pi(A)=\int_{E} \pi(\text{d}\bx) \changed{K}(\bx,A)$ for any $\bx\in E$ and $A \in \mathcal E$, then $\pi$ is called the {\it stationary distribution}.
See, for example, \citet{nummelin2004general} for the general theory of Markov chain\changed{s}.

The MCMC method is widely used for simulating a distribution by generating a Markov chain with the given distribution as a stationary distribution $\pi$.
For some distribution $\pi$ and $\pi$-measurable vector-valued function $\bm h$ on $E$, our estimand is denoted as
\begin{equation}\label{theoretical quantity}
{\bm \pi}({\bm h}):=\int_{E} {\bm h}(\bx)\pi(\text{d}\bx).
\end{equation}
The MCMC estimator of ($\ref{theoretical quantity}$) is given by 
\begin{equation}\label{mcmc estimator in general}
\hat{{\bm \pi}}_{N}(\bm h):=\frac{1}{N}\sum_{n=1}^{N}{\bm h}(\bX^{(n)}),
\end{equation} 
where $(\bX^{(1)},\dots,\bX^{(N)})$ is a sample path from time $1$ to $N$ (we call it an {\it $N$-path}) of a Markov chain whose stationary distribution is $\pi$.
The distribution $\pi$ is called the {\it target distribution}.
Since it is determined by the problem at hand, the problem is to find a stochastic kernel $K$ such that it has the stationary distribution $\pi$, and sample paths of its Markov chain can easily be generated.

One of the most popular stochastic kernels is the MH kernel defined by
\begin{equation*}\label{MH kernel}
K(\bx,{\rm d}\by) = k(\bx,\by){\rm d}\by + r(\bx)\delta_{\bx}(\by),
\end{equation*}
where
\begin{eqnarray*}
k(\bx,\by)&=&q(\bx,\by)\alpha(\bx,\by);\\
\alpha(\bx,\by) &=& \begin{cases}
\text{min}\left[ \frac{\pi(\by)q(\by,\bx)}{\pi(\bx)q(\bx,\by)}, \text{ }1\text{ } \right]&\text{ if } \pi(\bx)q(\bx,\by)>0,\\
0&\text{ otherwise};
\end{cases}\\
r(\bx) &=& 1- \int_{E} k(\bx, \by ){\rm d}\by;
\end{eqnarray*}
$\delta_{\bx}$ is the Dirac delta function; $q: E\times E \rightarrow \R_{+}$ is a function such that $\bx \mapsto q(\bx,\by)$ is measurable for any $\by \in E$; and $\by \mapsto q(\bx,\by)$ is a probability density for any $\bx \in E$.
This function $q$ is called a {\it proposal density}.
It can be shown that the MH kernel has stationary distribution $\pi$; see \citet{tierney1994markov}.
Under the three conditions (i)--(iii) where (i) at least one vector $\bx^{(0)}\in \text{supp}(\pi) $ is known, where $\text{supp}(\pi):=\{\bx \in E: \pi(\bx)>0\}$; (ii) samples  from $q(\bx,\cdot)$ can be generated for any $\bx \in E$; and (iii) the ratio $\pi(\by)/\pi(\bx)$ can be calculated for any $\bx,\by \in E$,  we can generate an $N$-path of the desired Markov chain by the following \emph{MH algorithm}:

\begin{flushleft}
{\bf Algorithm 1:} (MH algorithm)
\end{flushleft}
\begin{enumerate}\setlength{\itemsep}{-5pt}
\item[1.] Fix a sample size $N>0$, proposal density $q$, and initial value $\bm X^{(0)}=\bx^{(0)}\in \text{supp}(\pi)$.\\
\item[2.] {\bf For} $n=0,1,\dots,N-1$, do:\\
\item[3.] \hspace{10mm}Generate $\bX_{\ast}^{(n)}\sim q(\bX^{(n)},\hspace{1mm}\cdot\hspace{1mm})$ and $U\sim \mathcal U(0,1)$.\\
\item[4.] \hspace{10mm}Set 
\begin{equation}\label{acceptance probability}
\alpha_{n}:=\alpha(\bX^{(n)},\bX_{\ast}^{(n)})=\min \left[\hspace{2mm}
\frac{\pi(\bX_{\ast}^{(n)})q(\bX_{\ast}^{(n)},\bX^{(n)})}{\pi(\bX^{(n)})q(\bX^{(n)},\bX_{\ast}^{(n)})},
 \ 1\ \right].
\end{equation}
\item[5.]\hspace{10mm} Set 
\begin{equation*}
\bX^{(n+1)}:=1_{[U\leq \alpha_{n}]} \bX_{\ast}^{(n)}+ 1_{[U>\alpha_{n}]} \bX^{(n)}.\\
\end{equation*}
\item[6.] Return $(\bX^{(1)},\dots,\bX^{(N)})$. 
\end{enumerate}
We call $\alpha_{n}:=\alpha(\bX^{(n)},\bX_{\ast}^{(n)})$ in $(\ref{acceptance probability})$ the {\it acceptance probability} at the $n$th iteration.
Based on the $N$-path $(\bX^{(1)},\dots,\bX^{(N)})$ generated in Algorithm 1, the MCMC (MH) estimator $(\ref{mcmc estimator in general})$ is constructed.

Under regularity conditions, the MCMC estimator $\hat{{\bm \pi}}_{N}({\bm h})$ satisfies \emph{consistency} and the \emph{central limit theorem (CLT)}.
First, the MCMC estimator is consistent if
\begin{equation}\label{MCMC estimator consistency}
\lim_{N\rightarrow \infty }\hat{{\bm \pi}}_{N}({\bm h}) =  {\bm \pi}({\bm h})\quad \text{ a.s.},
\end{equation}
for any $\pi$-integrable function ${\bm h}$ and any initial state $\bX^{(0)}=\bx^{(0)} \in \text{supp}(\pi)$.
Next, CLT holds if
\begin{equation}\label{central limit theorem}
\sqrt{N}\{\hat{{\bm \pi}}_{N}({\bm h})- {\bm \pi}({\bm h})\} \darrow {\mathcal N}_{d}(\bzero,{\bf \Sigma}_{{\bm h}})  \quad \text{as} \quad N \rightarrow \infty,
\end{equation}
where the asymptotic variance matrix is given by
\begin{equation}\label{asymptotic variance}
{\bf \Sigma}_{{\bm h}}:=\text{Var}_{\pi}[{\bm h}(\bX^{(1)})]+2\sum_{k=1}^{\infty}\text{Cov}_{\pi}[{\bm h}(\bX^{(1)}),{\bm h}(\bX^{(k+1)})].
\end{equation}
Since the asymptotic variance $(\ref{asymptotic variance})$ can rarely be computed in a real situation, it is estimated from the sample path $(\bX^{(1)},\dots,\bX^{(N)})$ generated in Algorithm 1.
One popular estimator of ${\bf \Sigma}_{{\bm h}}$ is the so-called \emph{batch means estimator}; see \citet{geyer2011introduction}.
For an $N$-path $(\bX^{(1)},\dots,\bX^{(N)})$, the batch means estimator $\hat{{\bf \Sigma}}_{{\bm h},N}$ is defined by
\begin{equation}\label{batch mean estimator}
\hat{{\bf \Sigma}}_{{\bm h},N}=\frac{L_{N}}{B_{N}-1}\sum_{b=1}^{B_{N}}\{
\hat{{\bm \pi}}_{N,b}({\bm h}) -\hat{{\bm \pi}}_{N}({\bm h})  
\}\{
\tra{\hat{{\bm \pi}}_{N,b}({\bm h}) -\hat{{\bm \pi}}_{N}({\bm h})  
\}},
\end{equation}\label{mean over the batch}
where $L_{N}$ and $B_{N}$ are positive integers satisfying $N=L_{N} B_{N}$, and 
\begin{equation*}
\hat{{\bm \pi}}_{N,b}({\bm h})=\frac{1}{L_{N}}\sum_{l=(b-1)L_{N}}^{bL_{N}-1}{\bm h}(\bX^{(l)}) \quad \text{ for }\quad  b=1,2,\dots,B_{N}.
\end{equation*}
$L_{N}$ is called the \emph{batch length}, and $B_{N}$ is the number of batches.
Under regularity conditions, the batch means estimator $\hat{{\bf \Sigma}}_{{\bm h},N}$ converges to ${\bf \Sigma}_{{\bm h}}$ as $N \rightarrow \infty$; see \citet{jones2006fixed} and \citet{vats2015multivariate}.
By using CLT of $\hat{{\bm \pi}}_{N}({\bm h})$ and the consistency of $\hat{{\bf \Sigma}}_{{\bm h},N}$, one can construct an approximate confidence interval of the true quantity ${\bm \pi}({\bm h})$ based on an $N$-path of the Markov chain.

\subsection{Choice of the proposal distribution}\label{choice of the proposal distribution}
When implementing the MH, an appropriate choice of the proposal function $q$ is necessary since it affects the asymptotic variance \eqref{asymptotic variance}.
Since ${\bf \Sigma}_{h}$ can rarely be calculated explicitly in a real situation, a post-implementation review is usually conducted; that is, the goodness of the selected proposal distribution is evaluated after performing the MH.
In this section, we introduce two methods for evaluating the selected proposal distribution.
We also provide some families of proposal distributions for later use.

In practice, there are two prevalent methods to determine the performance of the proposal distribution. 
One is to inspect the {\it autocorrelation plots} of the marginal sample paths.
For an $N$-path $(\bX^{(1)},\dots,\bX^{(N)})$, vector-valued measurable function ${\bm h}(\bX)=\tra{(h_{1}(\bX),\dots,h_{d}(\bX))}$, and the MH estimator $\hat{{\bm \pi}}_{N}({\bm h})=\tra{(\hat{\pi}_{N,1},\dots,\hat{\pi}_{N,d})}$, the sample autocorrelations $\hat{r}_{j}(k):=\hat{R}_{j}(k)/\hat{R}_{j}(0)$ are drawn against the lag $k=0,1,2,\dots$, where
\begin{equation*}\label{autocovariance}
\hat{R}_{j}(k):=\frac{1}{N-k}\sum_{n=1}^{N-k}\{h_{j}(\bX^{(n)})-\hat{\pi}_{N,j}\}\{h_{j}(\bX^{(n+k)})-\hat{\pi}_{N,j}\},
\end{equation*}
for $j=1,2,\dots,d$.
From the form of the asymptotic variance $\eqref{asymptotic variance}$, one can expect that asymptotic variance ${\bf \Sigma}_{{\bm h}}$ is small if the autocorrelation plots steadily decline to zero as the lags increase.
Another implicative quantity is the {\it acceptance rate (ACR)}, which is the percentage of times a candidate $\bX_{\ast}$ is accepted through the whole run.
\changed{Since an appropriate ACR varies according to the shape and dimension of the target distribution, there is no general standard on the range of ACR; see \cite{rosenthal2011optimal} and references therein.}
Meanwhile, altering proposal distribution is generally suggested when extremely low or high ACR is observed.
\changed{Figure~\ref{Fig:MH:selection} illustrates two typical situations when such extreme ACRs are observed.
The first situation, highlighted in red in Figure~\ref{Fig:MH:selection}, represents the case when the chain tends to be stuck to one point due to, for example, a too high variance of the proposal distribution.
The resulting estimator can then have very high variance. 
The second situation, highlighted in blue, shows the case where the chain moves only around one mode of the target density and does not traverse the entire support of the target distribution.
This fallaciously high ACR yields high bias caused by ignoring other modes.
}

\begin{figure}[t]
\includegraphics[scale=0.25]{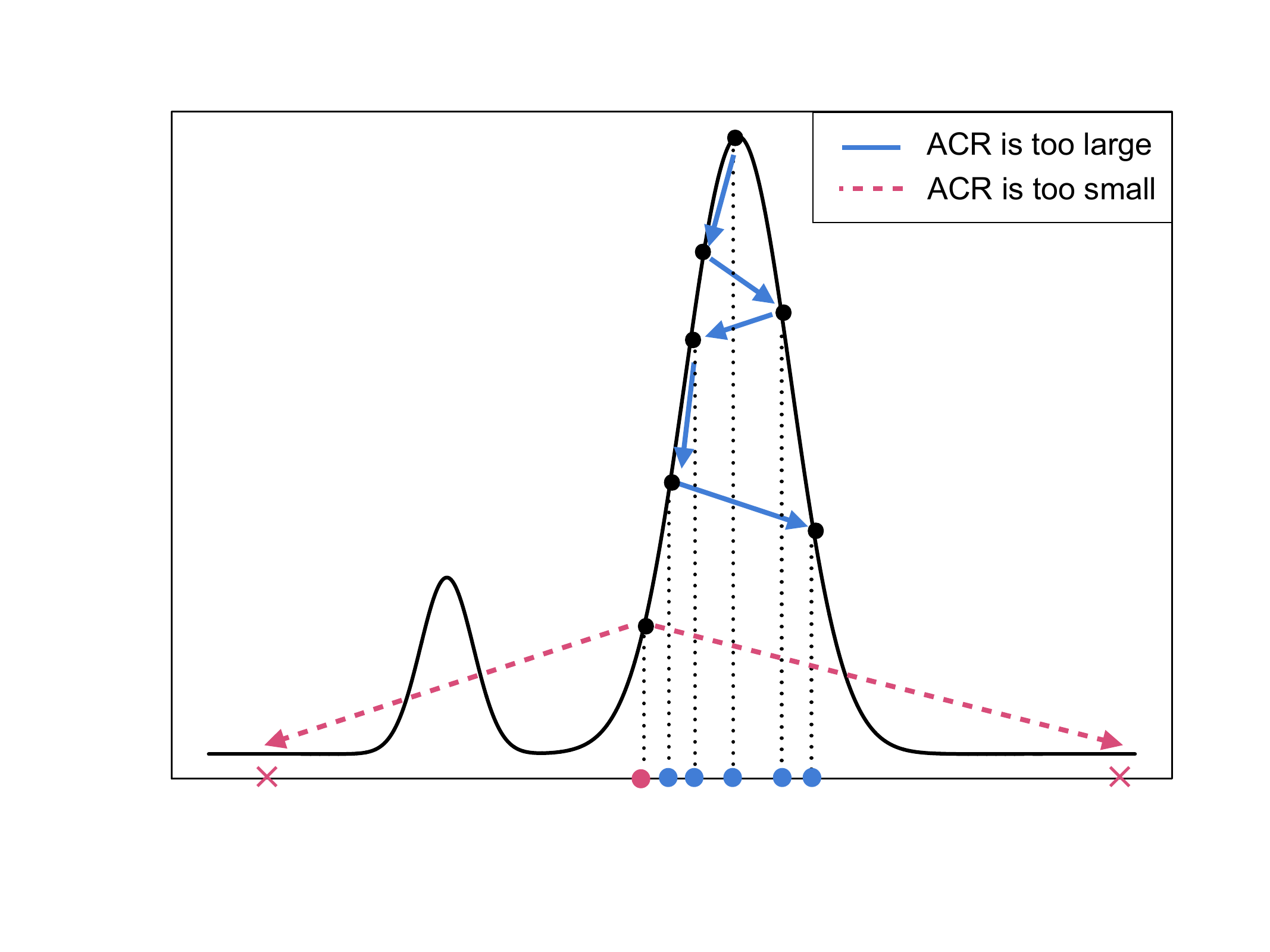}
\centering
\vspace{-10mm}
\caption{\changed{Two typical situations when an extremely low or high acceptance rate (ACR) of the Metropolis-Hastings algorithm is observed.
In the case highlighted in red, proposed candidates have very low values of the target density and the chain tends to be stuck to one point.
In the situation highlighted in blue, the chain only moves around one mode and does not traverse the entire support of the target distribution.
}}
\label{Fig:MH:selection}
\end{figure}

Typically, proposal distribution $q$ is selected among certain classes of distributions.
To find an appropriate $q$ depending on the target distribution, several classes of proposal distributions are presented in order.
First, if the proposal function is of the form $q(\bx,\by)=f(\by-\bx)$ for some density $f$, the candidate $\bX_{\ast}$ is drawn according to the following process: 
\begin{equation}\label{random walk proposal}
\bX_{\ast}=\bX+\bZ,\quad\text{where}\quad\bZ \sim f,
\end{equation}
and $\bX$ is the current state.
This type of $q$ is called the {\it random walk} proposal distribution.
In the case wherein $f$ is symmetric around the origin, the acceptance probability $(\ref{acceptance probability})$ is written simply as $\alpha(\bx,\by)=\min \left[ \frac{\pi(\by)}{\pi(\bx)},1\right]$.
Second, when $q(\bx,\by)=f(\by)$ for some density $f$, then candidate $\bX_{\ast}$ is updated by 
\begin{equation}\label{independent proposal}
\bX_{\ast}=\bZ,\quad\text{where}\quad\bZ \sim f.
\end{equation}
This $q$ is called the {\it independent} proposal distribution.
The two proposal distributions---random walk and independent---are widely used due to their simplicities.
However, these proposal distributions often fail to perform well when the target distribution $\pi$ is heavy-tailed.
To overcome this problem, the {\it mixed preconditioned Crank-Nicolson (MpCN)} proposal distribution is proposed by \citet{kamatani2014efficient}.
This proposal distribution updates the candidate according to the following process: 
\begin{equation}\label{MpCN proposal}
\bm X_{\ast}=\bmu+\rho^{\frac{1}{2}} (\bX-\bmu)+(1-\rho)^{\frac{1}{2}} Z^{-\frac{1}{2}}\cdot{\bm W},
\end{equation} 
where $\rho \in (0,1)$, $Z$ follows the gamma distribution with shape parameter $d/2$ and scale parameter $||{\bf \Sigma}^{-\frac{1}{2}}(\bX-\bmu)||^{2}/2$, and $\bW \sim {\mathcal N_{d}}(\bzero,{\bf \Sigma})$ for some $d$-vector $\bmu\in \R^{d}$ and $d\times d$ matrix ${\bf \Sigma} \in {\mathcal M}^{d\times d}_{+}$\changed{, where ${\mathcal M}^{d\times d}_{+}$ is the set of all positive definite $d\times d$ matrices}.
Throughout this paper, $\rho$ is set to be 0.8 as a default choice in \citet{kamatani2014efficient}.
Ideally, $\bmu$ and ${\bf \Sigma}$ are set to be $\bmu=\E[\bX]$ and ${\bf \Sigma}=\Var[\bX]$, while in practice, they can be replaced by their rough estimates since moments of $\bX$ are typically unknown.
Note that the original MpCN proposed in \citet{kamatani2014efficient} is the standardized version (that is, $\bmu=\bzero$ and ${\bf \Sigma}=\bI_{d}$, where $\bI_d$ is an identity matrix).
The acceptance probability ($\ref{acceptance probability}$) of the MpCN proposal distribution can be written as
\begin{equation*}\label{acceptance probability of MpCN}
\alpha(\bX,\bX_{\ast})=\left[\hspace{1mm}
\frac{\pi(\bX_{\ast})}{\pi(\bX)} 
\left(
\frac{
  ||{\bf \Sigma}^{-\frac{1}{2}}(\bX-\bmu)||
  }{
  ||{\bf \Sigma}^{-\frac{1}{2}}(\bX_{\ast}-\bmu)||
  }
\right)^{-d},1\  \right].
\end{equation*}
One of the key differences between this proposal distribution and the first two simple ones is that in the MpCN, not only the mean but also the variance of the candidate changes with the current state $\bX$.
Since the MpCN proposal distribution admits larger jumps in the tail, a better acceptance rate can be expected even when $\pi$ is heavy-tailed.

\section{The proposed method}\label{the proposed method}

In this section, we propose a new estimator of VaR contributions that utilizes the MCMC method, especially the MH algorithm, to achieve an efficient estimation.
Theoretical study on the consistency and asymptotic normality of our MH-based estimator is provided in the Appendix for certain classes of risk models.

\subsection{Assumptions and setup}\label{the set-up}
We start by declaring assumptions under which our MH estimator is applicable.
\begin{assumption}
On applying the MH estimator, we suppose the following:
\begin{enumerate}\setlength{\itemsep}{-4mm}
\item[(i)] an explicit form of the joint loss density $f_{\bX}$ is given, and thus one can compute the quantity $f_{\bX}(\bx)$ for any $\bx \in \R^{d}$;\\
\item[(ii)] a generator of i.i.d. samples from the loss distribution  $F_{\bX}$ is available; and\\
\item[(iii)] neither the explicit form of total loss density $f_{S}$ nor the way to compute the quantity $f_{S}(\text{VaR}_{p}(S))$ is available. 
\end{enumerate}
\end{assumption}
Note that assumption 1 (ii) enables us to generate samples from $F_{S}$ by setting $S^{(n)}=X_{1}^{(n)}+\cdots+X_{d}^{(n)}$ where ($X_1^{(n)},\dots,X_d^{(n)}$) is an $n$th sample from $F_{\bX}$.
\changed{Assumption 1 (iii) implies that, while $f_S$ can be derived from $f_{\bX}$ in the form of an integral, the integral is not straightforward to calculate.}
Such a situation typically occurs when the joint loss density $f_{\bX}$ is specified through a copula density $c$ and marginal loss densities $\sqc{f}{d}$.
The resulting joint loss density $f_{\bX}$ is specified as in formula $(\ref{sklar theorem density form})$.

As is mentioned in section~\ref{capital allocation problem}, computing VaR contributions involves two steps; the first is to estimate VaR$_{p}(S)$, and the second is to estimate VaR contributions AC$=\E[\bX|S=\text{VaR}_{p}(S)]$ with VaR$_{p}(S)$ replaced by its estimate.
The estimation of VaR$_{p}(S)$ in the first step is often conducted with an MC simulation.
Based on i.i.d. samples ($S^{(1)},\dots,S^{(N)}$) from $F_{S}$, VaR$_{p}(S)$ can be estimated, for example, by $\widehat{\text{VaR}}_{p}(S)=S^{\lceil Np \rceil}$, where $\lceil Np \rceil$ is the smallest integer greater than $Np$, and $S^{\lceil Np\rceil}$ is the $\lceil Np\rceil$th largest sample among $N$ samples.
Since $\widehat{\text{VaR}}_{p}(S)$ is a deterministic quantity, one can regard it as a constant $v = \widehat{\text{VaR}}_{p}(S)$.

In the second step, AC$=\E[\bX|S=v]$ is estimated.
According to the crude MC method, VaR contributions are estimated by \eqref{MC estimator}.
As explained in section \ref{capital allocation problem}, the problem of this two-step procedure is that the estimator of VaR contributions in the second step is typically biased.
To address this issue, we develop an MCMC (MH)-based estimator that achieves consistency and high sample efficiency.

\subsection{The MH estimator of VaR contributions}\label{subsec: the MCMC estimator of VaR contributions}
We propose to estimate VaR contributions by sequentially updating samples so that all samples lie in the set $\mathcal S_v = \{\bx \in \R^{d}:x_{1}+\cdots+x_{d}=v\}$.
The updating rule is established so that the componentwise sum of each sample is preserved and the samples are taken from the distribution $F_{\bX|S=v}$.
We start to describe the MH-based estimator by reformulating the problem of computing VaR contributions.
\changed{Since MH requires a density of the target distribution but the $d$-dimensional density is not well-defined on the degenerated space $\mathcal S_v$, we consider the first $d'$ losses $\bX'=\tra{(\sqc{X}{d'})}$ where $d'=d-1$.}
By the full allocation property $(\ref{full allocation property})$, it holds that
\begin{equation*}
\E[\bX|S=v]=\tra{(\E[\bX'|S=v],\hspace{1mm}v-\tra{\bone_{d'}} \E[\bX'|S=v])},
\end{equation*}
where $S=X_{1}+\cdots+X_{d}$.
Therefore, computation of VaR contributions AC=$\E[\bX|S=v]$ can be reduced to estimate VaR contributions of the $d'$-subportfolio, denoted by AC$'=\E[\bX'|S=v]$.
In our method, this quantity AC$'$ is estimated by generating samples directly from $F_{\bX'|S=v}$.
The conditional joint density of $\bX'$ given $\{S=v\}$ can be written as
\begin{equation*}\label{conditional density reformulation}
f_{\bX'|S=v}(\bx')=\frac{f_{\bX',S}(\bx',v)}{f_{S}(v)}
=\frac{f_{\bX}(\bx',v-\tra{\bone_{d'}}\bx')}{f_{S}(v)},\quad  \bx' \in \R^{d'},
\end{equation*}
where the last equation follows from a linear transformation $\tra{(\bX',S)}\mapsto \bX$.
At this point, sampling directly from $f_{\bX'|S=v}$ is difficult since the total loss density $f_{S}(v)$ is not easy to evaluate.

By taking $E=\R^{d'}$, $h(\bx)=\bx$, and $\pi(\bx)=f_{\bX'|S=v}(\bx)$ in the notations presented in subsection $\ref{brief introduction to MCMC}$, our problem of estimating VaR contributions can be reduced to estimate ${\bm \pi}({\bm h})=\E[\bX'|S=v]$ in ($\ref{theoretical quantity}$) by MCMC.
Even though it is challenging to compute $f_{\bX'|S=v}$, we can compute the acceptance probability ($\ref{acceptance probability}$) given by
\begin{equation*}\label{reduced acceptance probability}
\alpha(\bx,\by)=
\min \left[
\frac{f_{\bX'|S=v}(\by) q(\by,\bx)}{f_{\bX'|S=v}(\bx)  q(\bx,\by)},
 \ 1\right]
 =\min \left[
\frac{
f_{\bX}(\by,v-\tra{\bone_{d'}}\by)
q(\by,\bx)
}{
f_{\bX}(\bx,v-\tra{\bone_{d'}}\bx)
q(\bx,\by)}
, \ 1\right],
\end{equation*}
for any $\bx,\by$ by assumption 1 (i).
Note that the term $f_{S}(v)$ disappears by taking the ratio of $f_{\bX'|S=v}(\by)$ to $f_{\bX'|S=v}(\bx)$.
Therefore, under an appropriate choice of the proposal density $q$, the MH algorithm (algorithm 1) allows one to generate $N$-paths of the Markov chain whose stationary distribution is $\pi(\bx)=f_{\bX'|S=v}(\bx)$.
Based on a sample path, we can construct the MH estimator $\hat{{\bm \pi}}_{N}({\bm h})$ defined by $(\ref{mcmc estimator in general})$.
The algorithm to compute the MH estimator of VaR contributions is summarized in the following algorithm.
\begin{flushleft}
{\bf Algorithm 2:} (MH estimator of VaR contributions $\E[\bX|S=\text{VaR}_{p}(S)]$)
\end{flushleft}
\begin{enumerate}\setlength{\itemsep}{0mm}
\item[1.] Estimate VaR as $v=\widehat{\text{VaR}}_{p}(S)$ by MC samples.
\item[2.] Fix the sample size $N>0$, proposal distribution $q$, and initial value $\bX^{(0)}=\bx^{(0)} \in \text{supp}(f_{\bX'|S=v})$.
\item[3.] Perform Algorithm 1 (MH) for the given $N$, $q$, and $\bx^{(0)}$ to generate an $N$-path $(\bX'^{(1)},\dots,\bX'^{(N)})$. 
\item[4.] Set
\begin{equation}\label{MCMC estimator}
\estAC{MCMC}{q,N}=\frac{1}{N}\sum_{n=1}^{N}\bX^{(n)}\quad \text{where} \quad
\bX^{(n)}:=\tra{(\bX'^{(n)},v-\tra{\bone_{d'}}\bX'^{(n)})},
\end{equation}
to estimate VaR contributions AC$=\E[\bX|S=v]$.
\end{enumerate}

\changed{Note that Algorithm 2 can be easily extended to other choices of the function $h$, that is, one can estimate $\E[h(\bX)|S=v]$ for general functions $h$ by replacing \eqref{MCMC estimator} with $\frac{1}{N}\sum_{n=1}^{N}h(\bX^{(n)})$.}
Moreover, under regularity conditions, consistency and asymptotic normality of the MH estimator~\eqref{MCMC estimator} hold; see Appendix~\ref{sec: consistency and asymptotic normality} for more detail.

\changed{
\begin{remark}[More advanced MH algorithms]
On estimating VaR contributions, the MH algorithm can be extended to other MCMC methods.
To be precise, most MCMC methods, such as \emph{Gibbs sampler} \citep{geman1984stochastic,gelfand1990sampling} and \emph{Metropolis-adjusted Langevin algorithm} \citep[\emph{MALA,}][]{roberts1996exponential}, are special cases of MH by choosing the proposal distribution $q$ as a function of the target distribution $\pi$.
The Gibbs sampler incorporates the \emph{full conditional distributions}
$\pi_{j}(x_j|\bx_{-j}):=\pi(\bx)/\int_{\R}\pi(\bx)\dd x_{j}$ where $\bx_{-j}=(x_1,\dots,x_{j-1},x_{j+1},\dots,x_{d'})$ for $j=1,\dots,d'$ into $q$ and the MALA utilizes the gradient $\nabla \text{log}\pi(\bx)$ to define $q$.
For our target density $\pi(\bx)=f_{\bX|S=\text{VaR}_p(S)}$, the full conditional distributions $\pi_{j}(x_j|\bx_{-j}),\ j=1,\dots,d'$, are rarely available. 
On the other hand, $\nabla \text{log}\pi(\bx)$ is available when $\frac{\partial}{\partial x_j}f_{\bX}(x_1,\dots,x_{d'},v-\tra{\bone_{d'}}\bx)$, $j=1,\dots,d'$ can be computed for any $\bx \in \R^{d'}$ since the term $f_{S}(v)$ disappears when calculating $\frac{\partial}{\partial x_j}\text{log}\pi(\bx)$.
Therefore, the Gibbs sampler is rarely available but the MALA can be a possible choice if partial derivatives of $f_{\bX}$ are straightforward to compute.
\end{remark}
}

\begin{remark}[MCMC methods for ES contributions]
\changed{The MCMC methods can potentially be applied to estimate other risk contributions such as ES contributions~\eqref{ES contributions}.
For the corresponding target distribution $\pi(\bx)=f_{\bX|S\geq \text{VaR}_p(S)}(\bx), \bx \in \R^d$, the Gibbs sampler is sometimes feasible since its $j$-th full conditional density is given by
$\pi_{j}(x_j|\bx_{-j})=f_{X_j|\{X_j \geq \text{VaR}_p(S)-\tra{\bone_{d-1}}\bx_{-j}, \bX_{-j}=\bx_{-j}\}}(x_j)$ and this can be simulated if $j$-th conditional copula of $\bX$ is available, which is the case for Gaussian, $t$- and (survival) Clayton copulas; see, for example, \citet{cambou2017quasi}.
}

\changed{
While such an MCMC method is applicable to estimate ES contributions, we do not consider it in this paper for several reasons. 
One reason is that, compared with the crude MC, significant computational improvement cannot be expected in general by introducing MCMC.
The crude MC with sample size $N$ still retains $N(1-p)$ samples from $\bX$ to estimate AC$^{\text{ES}}=\E[\bX|S\geq \text{VaR}_p(S)]$ since $\Prob(S\geq \text{VaR}_p(S))=1-p$.
Moreover, from the regulatory viewpoint, a lower confidence level ($p=0.9\text{--}0.95$) of ES is plausible compared with that of VaR.
Therefore, the crude MC method still performs with enough accuracy.
Another reason is that simulation of the Gibbs sampler is typically slow especially when $d$ is large since the Gibbs sampler updates the chain componentwise.
Faster methods are often accessible in the crude MC for large dimensional portfolios.
}
\end{remark}

\section{Numerical experiments}\label{Simulation and empirical studies}
In this section, we apply the MH estimator proposed in section $\ref{the proposed method}$ to various risk models, and compare its performance with other existing estimators of VaR contributions.
Our simulation and empirical study based on real-world data show that the MH estimator has smaller bias and lower MSE compared with other estimators for many situations, including high-dimensional ($d\approx 500$) cases.
Based on the numerical experiments, we provide practical guidelines on how to choose an appropriate proposal distribution of the MH estimator given a risk model.
In the experiments, we used a MacBook Air, 1.4 GHz Intel Core i5, 4 GB 1600 MHz DDR3.

\subsection{Simulation study}\label{subsec: simulation study}

\subsubsection{Description of the numerical comparison}\label{description of the numerical comparison}

In the simulation study, we consider four risk models that are modelled separately by marginal densities and copula density.
We adopt heavy-tailed marginal distributions and copulas with upper-tail dependences as they are often of concern in risk management.
In all risk models, we set the size of the portfolio $d=3$.
The models are set as follows.
\begin{enumerate}\setlength{\itemsep}{-0mm}
\item[(1)] The loss random variables $X_{1}, X_{2}$, and $X_{3}$ follow homogeneous Pareto distributions \eqref{pareto density}
for $\kappa=4$ and $\gamma=3$.
The loss random vector $(X_{1},X_{2},X_{3})$ has a $d$-dimensional survival Clayton copula \eqref{survival clayton copula}
with $\theta=0.5$.\\
\item[(2)] $X_{1},X_{2}$, and $X_{3}$ have the same marginal distributions as in case (1).
Their copula is a student's $t$-copula~\eqref{student t copula density}
with the degree of freedom $\nu=4$, and the dispersion matrix ${\bf P}$ given by
\begin{align}
{\bf P}=\begin{pmatrix}\label{correlation matrix}
1 & -0.5 & 0.3\\
-0.5 & 1 & 0.5\\
0.3 & 0.5 & 1{}
\end{pmatrix}.
\end{align}
\item[(3)] $X_{1},X_{2}$, and $X_{3}$ follow homogeneous student's $t$-distribution with degree of freedom $\nu=4$, location parameter $\mu=0$, and scale parameter $\sigma=1$.
$(X_{1},X_{2},X_{3})$ has a survival Clayton copula \eqref{survival clayton copula} with $\theta=0.5$.\\
\item[(4)] $(X_{1},X_{2},X_{3})$ follows a multivariate student's $t$-distribution~\eqref{multivariate student t density}
 with $\nu=4$, $\bmu=\bzero$, and ${\bf \Sigma}={\bf P}$ where ${\bf P}$ is defined in $(\ref{correlation matrix})$.
\end{enumerate}
The first two models (1) and (2) consider pure losses, and the last two, (3) and (4), consider Profit\&Loss.
In all models, marginal distributions have variances of 2.0 and heavy tails with tail indices $5.0$.
Models (1) and (3) possess homogeneous upper-tail dependences with tail coefficients $\lambda^{U}=0.025$; see \citet[][]{joe2014dependence} for formulas on the tail coefficients.
Models (2) and (4) have upper, lower, and upper-lower tail dependences with tail coefficients $\lambda_{1,2}^{U}=\lambda_{1,2}^{L}=0.012$, $\lambda_{1,3}^{U}=\lambda_{1,3}^{L}=0.162$, $\lambda_{2,3}^{U}=\lambda_{2,3}^{L}=0.253$, $\lambda_{1,2}^{UL}=\lambda_{1,2}^{LU}=0.253$, $\lambda_{1,3}^{UL}=\lambda_{1,3}^{LU}=0.029$, and $\lambda_{2,3}^{UL}=\lambda_{2,3}^{LU}=0.012$.
As inferred by the dispersion matrix $(\ref{correlation matrix})$, the first and second losses are negatively dependent, while other pairs of losses are positively dependent.

For each risk model, we compute several estimators of VaR contributions $\text{AC}=\E[\bX|S=\text{VaR}_{p}(S)]$ for a confidence level $p=0.999$ with the $\text{VaR}_{p}(S)$ replaced by its Monte Carlo estimate $v=S^{[Np]}$. \changed{Note that $p=0.999$ is one of the standard confidence levels of VaR under the Basel capital regulation for various risks; see, for example, \citet{zimper2014minimal} and \citet{bcbs2011operisk}.}
The estimators we compare are the MC estimator $(\ref{MC estimator})$, 
NW estimator $(\ref{nw estimator})$ \citep{tasche2009capital}, GR estimator $(\ref{gr estimator})$ \citep{hallerbach2003decomposing,tasche2004approximations}, and the MH estimator $(\ref{MCMC estimator})$:
\begin{equation*}
\begin{array}{l@{\hspace{10mm}}l}
\estAC{MC}{\delta,N}=
\frac{1}{M_{\delta,N}}\sum_{n=1}^{N} \bX^{(n)}  1_{[S^{(n)}\in A_{\delta}]},&
\estAC{NW}{\phi,h,N}=\frac{
    \sum_{n=1}^{N} \bX^{(n)}  \phi\left(\frac{S^{(n)}-v}{\Delta}\right)
}{
  \sum_{n=1}^{N}\phi\left(\frac{S^{(n)}-v}{\Delta}\right)
},
\\[+10mm]
\estAC{GR}{g_{\bbeta},N}=g_{\hat{\bbeta}_{N}}(v), &
\estAC{MCMC}{q,N}=\frac{1}{N}\sum_{n=1}^{N}\bX^{(n)}_{|S=v},
\end{array}
\end{equation*}
where $\bX^{(1)},\dots,\bX^{(N)} \iidsim F_{\bX}$, $S^{(n)}:=X_{1}^{(n)}+\cdots+X_{d}^{(n)}$ and $(\bX^{(1)}_{|S=v},\dots,\bX^{(N)}_{|S=v})$ is an $N$-path of a Markov chain whose stationary distribution is $F_{\bX|S=v}$.
\changed{Note that the IS estimator~\eqref{IS estimator} is not reported since we could not find instrumental distributions which provide reasonably small variances of the estimators for the risk models above.}

For all estimators, we fix the sample size  $N=10^6$.
Other parameters of the estimators are determined as follows.
First, for the MC estimator, we set $\delta>0$ such that the MC sample size $M_{\delta,N}$ is around $10^3$.
For a fixed $\delta$, asymptotic normality
\begin{equation*}
\sqrt{M_{\delta,N}}\cdot(\estAC{MC}{\delta,N}-\text{AC}_{\delta})\rightarrow {\mathcal N}_{d}(\bzero,{\bf \Sigma}_{\delta}), \quad\text{as}\quad N \rightarrow \infty
\end{equation*}
holds; see, for example, \citet[][]{glasserman2013monte}.
We report the estimate of $\estAC{MC}{\delta,N}$ and its approximated standard error ${\bm S}_{MC}^{(j,j)}/\sqrt{M_{\delta,N}}$ for $j=1,2,3$, where ${\bm S}_{MC}^{(i,j)}$ is the $(i,j)$-component of the sample standard deviation ${\bm S}_{MC}$ defined by
\begin{align*}\label{sample variance matrix over MC samples}
{\bm S}_{MC}=\sqrt{
  \frac{1}{M_{\delta,N}}
  \sum_{n=1}^{N} \left(\bX^{(n)}-\estAC{MC}{\delta,N}\right)\tra{\left(\bX^{(n)}-\estAC{MC}{\delta,N}\right)}  1_{[S^{(n)}\in A_{\delta}]}
}.
\end{align*}
Second, for the NW estimator, we choose the kernel density $\phi$ as the standard normal density.
We decide the bandwidth $\Delta=1.06 \hat{\sigma}_{S}   N^{-1/5}$ according to Silverman's rule of thumb \citep{pagan1999nonparametric}.
Although asymptotic normality holds for the NW estimator, its asymptotic variance can hardly be computed because it requires the evaluation of $f_{S}(v)$.
Therefore, we report only the estimate of $\estAC{NW}{\phi,h,N}$.
Third, for the GR estimator, we choose $g_{\bbeta}(s)=\beta_{0}+\beta_{1} s$, and its coefficients are estimated by
\begin{equation*}\label{GR beta estimator}
\hat{\bbeta}_{N,1}=\frac{\sum_{n=1}^{N}(\bX^{(n)}-\bar{\bX}_{N})(S^{(n)}-\bar{S}_{N})
}{
\sum_{n=1}^{N}(S^{(n)}-\bar{S}_{N})^{2}
},\quad
\hat{\bbeta}_{N,0}=\bar{\bX}_{N}-\hat{\bbeta}_{N,1} \bar{S}_{N},
\end{equation*}
where $\bar{\bX}_{N}=\frac{1}{N}\sum_{n=1}^{N}\bX^{(n)}$ and $\bar{S}_{N}=\frac{1}{N}\sum_{n=1}^{N}S^{(n)}$.
Under regularity conditions, the asymptotic normality
\begin{equation*}
\sqrt{N} \left\{ 
\begin{pmatrix}
\hat{\beta}_{N,0}^{(j)}  \\[5pt]
\hat{\beta}_{N,1}^{(j)}  \\
\end{pmatrix}-
\begin{pmatrix}
\beta_{0}^{(j)}  \\[5pt]
\beta_{1}^{(j)}  \\
\end{pmatrix}
\right\}
\rightarrow N_{2}(\bzero,\sigma_{\ep_{j}}^{2}{\bf Q}^{-1}),
\quad \text{ as } \quad N \rightarrow \infty
\end{equation*}
holds for $j=1,2,3$, where $\hat{\beta}_{N,k}^{(j)}$ and $\beta_{k}^{(j)}$ are the $j$th components of $\hat{\bbeta}_{N,k}$ and $\bbeta_{k}$, respectively, for $k=1$ and $2$; $\ep_{j}$ is the $j$th component of the error term $\bep$; $\sigma_{\ep_{j}}^{2}$ is the conditional variance of $\ep_{j}$ given $S_{(1)},\dots,S_{(N)}$; and 
\begin{equation*}
{\bf Q}:=\lim_{N\rightarrow \infty}\tra{{\bf Y}}{\bf Y}/N, \quad\text{where}\quad
{\bf Y}=\tra{\left(
\begin{array}{lll}
1 & \cdots &1\\
S^{(1)} & \cdots &S^{(N)}
\end{array}
\right)}.
\end{equation*}
Based on these results, we report the estimate of $\estAC{GR}{g_{\bbeta},N}$ and its approximated standard error ${\bf S}_{GR}^{(j)}/\sqrt{N}$ for $j=1,2,3$, where
\begin{equation*}
{\bf S}_{GR}^{(j)}:=\sqrt{\hat{{\bf \Sigma}}_{GR,j}^{(1,1)}+2 v  \hat{{\bf \Sigma}}_{GR,j}^{(1,2)}+v^{2} \hat{{\bf \Sigma}}_{GR,j}^{(2,2)}},
\end{equation*}
$\hat{{\bf \Sigma}}_{GR,j}=\hat{\sigma}_{\ep_{j}}^{2}\cdot(\tra{{\bf Y}}{{\bf \bY}}/N)^{-1}$, and $\hat{\sigma}_{\ep_{j}}$ is the sample standard deviation of the $j$th residuals.
Finally, for the MH estimator, we choose different proposal distributions depending on risk models (1)--(4).
For each risk model, we choose (1) a random walk proposal $q(\bx,\by)=f(\by-\bx)$ with $f \sim \mathcal N_{d}(\bzero,\hat{{\bf \Sigma}}_{v})$, where $\hat{{\bf \Sigma}}_{v}:={\bf S}^{2}_{MC}$; (2) an independent proposal $q(\bx,\by)=f(\by)$, where $f$ is the density of the Dirichlet distribution with parameters \changed{$(0.323, 0.448, 0.892)$, which is estimated by the maximum likelihood method from pseudo samples generated by MC}, (3) and (4) MpCN proposal with $\rho=0.8$, $\bmu=\tra{(\estAC{GR}{g_{\bbeta},N}\text{}',v-\tra{\bone}_{d'}\estAC{GR}{g_{\bbeta},N}\text{}')}$, and ${\bf \Sigma}:={\bf S}^{2}_{MC}$.
In Algorithm 2, we set the initial values as $\bx^{(0)}=\tra{(v/3,v/3,v/3)}$.
We estimate the asymptotic variances of MH estimators by the batch means estimators $\hat{{\bf \Sigma}}_{N}$ defined by $(\ref{batch mean estimator})$.
Following \citet{jones2006fixed}, we choose $L_{N}:=\lfloor N^{\frac{1}{2}} \rfloor=10^3$ and $B_{N}:=\lfloor N/L_{N} \rfloor=\lfloor N^{\frac{1}{2}} \rfloor=10^3$.
We report the estimate of $\estAC{MCMC}{q,N}$ and its approximated standard error $\hat{{\bf \Sigma}}_{N}^{(j,j)}/\sqrt{N}$ for $j=1,2,3$.

\begin{landscape}
\centering
\begin{figure}[h]
\includegraphics[scale=0.85]{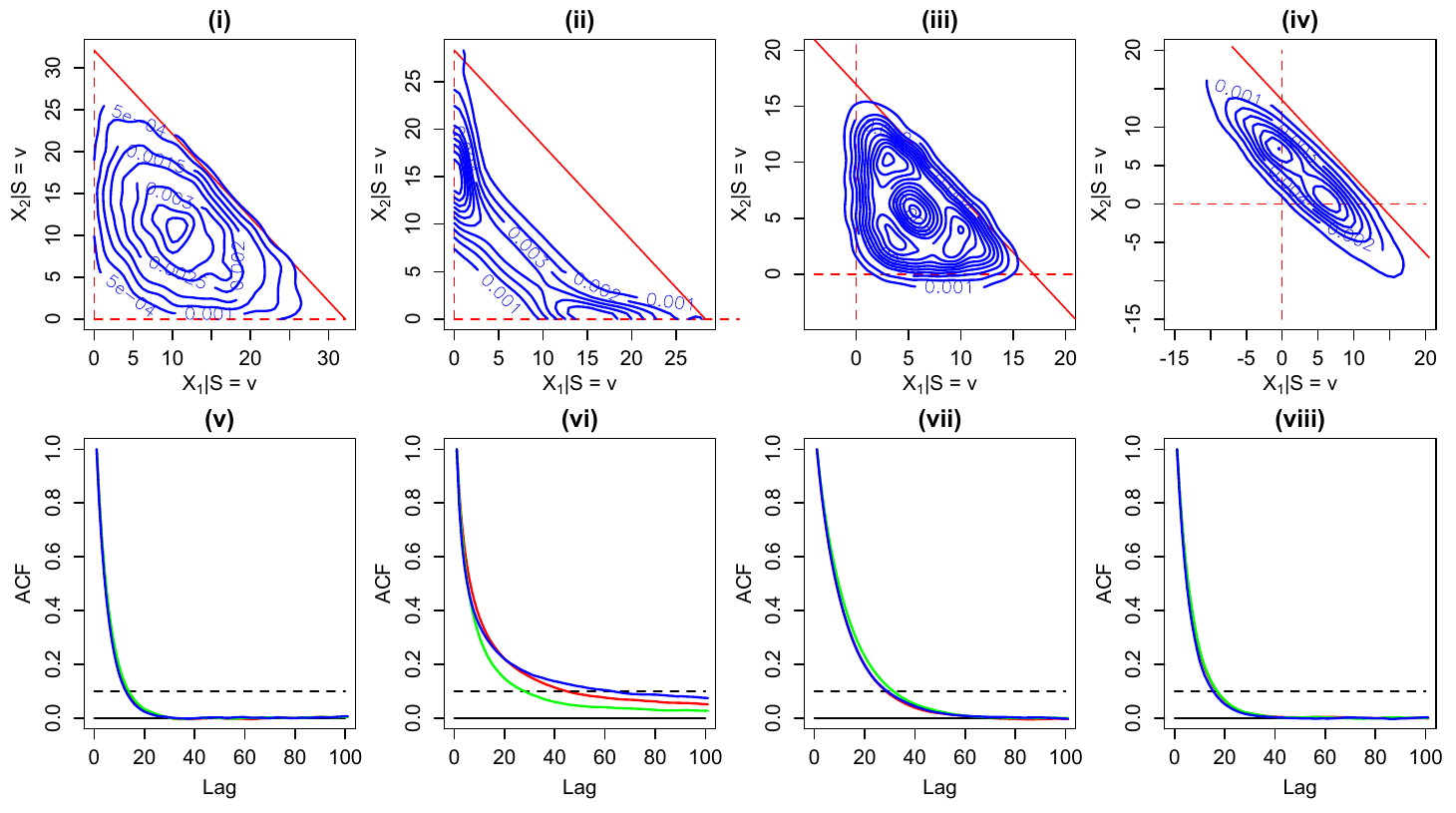}\vspace{0mm}
\caption{Contour plots (i)--(iv) and autocorrelation plots (v)--(viii) of Markov chains generated by Algorithm 2 for four different risk models:
(i) and (v) Pareto + Clayton; 
(ii) and (vi) Pareto + $t$-copula; 
(iii) and (vii) Student's $t$ + Clayton; and 
(vi) and (viii) Student's $t$ + $t$-copula. 
The red lines represent the edges of the $v$-simplex, where $v$ is the estimate of VaR$_{p}(S)$.
The dotted black lines in plots (vi)--(viii) represent $y=0.1$.
When drawing the contour plots, we used subsamples that are picked up every 100th point of the original Markov chains to reduce dependence among samples. 
}
\label{QF_Figure1}
\end{figure}
\end{landscape}

\subsubsection{Results and discussions}\label{results and discussions}
Due to the simplicity of MC, NW, and GR estimators, they were calculated instantly for all risk models.
\changed{On the other hand, MH estimators took (1) 2.324, (2) 1.425, (3) 3.828, and (4) 3.698 minutes to generate an $N$-path and to compute the estimators.}
As mentioned in section $\ref{choice of the proposal distribution}$, the validity of the proposal selection can be inspected by autocorrelation plots and ACR.
 Figure $\ref{QF_Figure1}$ (v)--(viii) shows the autocorrelation plots of the Markov chains generated by Algorithm 2.
The acceptance rate of the MH algorithm in each risk model was (1) 0.566, (2) 0.222, (3) 0.604, and (4) 0.767.
In  Figure $\ref{QF_Figure1}$ (v)--(viii), we can observe that the autocorrelation plots steadily decline below 0.1 by lag $h$ around 100 for all risk models.
Together with the observations that the ACRs are moderate, we could state that the choices of the proposal distributions above are appropriate for all risk models.

Before showing the results of the estimation, let us check the shapes of the conditional distributions $F_{\bX'|S=v}$ by plotting the $N$-path generated by Algorithm 2.
Figure $\ref{QF_Figure1}$ (i)--(iv) shows the contour plots of the generated Markov chains.
According to these plots, the features of the conditional distribution $F_{\bX'|S=v}$ in each risk model are summarized as follows:
\begin{enumerate}\setlength{\itemsep}{+1mm}
\item[(1)] {\it Pareto + survival Clayton:} The contour plot in Figure $\ref{QF_Figure1}$ (i) shows that $F_{\bX'|S=v}$ has a unique mode.
The density steadily decays as they move away from the mode.
In addition, the contour plot seems symmetric at the diagonal line.

\item[(2)]  {\it Pareto + $t$-copula:} Unlike case (1), $F_{\bX'|S=v}$ seems to possess two distinct modes close to the axes.
High probabilities are concentrated around the edges of the simplex.
Along with the axes, the gradients of the density seem to be sharp.
Moreover, the contour plot in Figure $\ref{QF_Figure1}$ (ii) is asymmetric at the diagonal line $y=x$. 

\item[(3)]  {\it Student's $t$ + survival Clayton:} Although the conditional loss random vector $\bX'|S=v$ can take negative values, it is supported mostly on the bounded simplex as in case (1).
The contour plot in Figure $\ref{QF_Figure1}$ (iii) seems unimodal and symmetric around the diagonal.
The tails of $F_{\bX'|S=v}$ are obviously light.

\item[(4)]  {\it Student's $t$ + $t$-copula:} In this case, the conditional distribution $F_{\bX'|S=v}$ can be shown to be a \emph{Pearson type VII distribution} $(\ref{pearson type vii})$ in Appendix~\ref{sec: consistency and asymptotic normality}.
From the contour plot in Figure $\ref{QF_Figure1}$ (iv), we can observe elliptical symmetry and tail-heaviness.
Unlike case (3), the loss vector $\bX'|S=v$ can take large negative values beyond the bounded simplex.
\end{enumerate}

\begin{table}[t]\centering
\begin{threeparttable}
\caption{Estimates (biases) and standard errors (rooted mean squared errors; RMSEs) of the four different estimators of value-at-risk contributions under four different risk models\tnote{$\dagger$}. \changed{The best result in each risk model is highlighted in bold.} }\vspace{-2mm}
\label{Table 1}\vspace{-2mm}
\begin{tabular}{l@{\hspace{0mm}}r@{\hspace{4mm}}r@{\hspace{4mm}}r@{\hspace{4mm}}r@{\hspace{4mm}}r@{\hspace{0mm}} c@{\hspace{6mm}} r@{\hspace{4mm}}r@{\hspace{4mm}}r} \toprule[1.5pt] \\[-7pt]
&& \multicolumn{4}{c}{{\it Estimate of AC (Bias): }} &  &
  \multicolumn{3}{c}{{\it Standard error ($\sqrt{MSE}$): }}  \\[6pt] \cline{3-6} \cline{8-10} \\[-4pt]
&{\it Estimator} &  $\bm{MC}$ & $\bm{NW}$ & $\bm{GR}$ & $\bm{MH}$ && $\bm{MC}$ & $\bm{GR}$ & $\bm{MH}$ \\[6pt]
  \hline \\[-2pt]
{\it (1)}&\multicolumn{9}{l}{{\it Pareto + survival Clayton: True AC =  (10.708, 10.708, 10.708)} }\\[5pt]
&AC$_{1}$ &10.575 & 11.744 & 10.745 & \bf{10.708} && 0.173 & \bf{0.008 }& 0.019 \\[2pt] 
&   & (-0.133) & (1.036) & (0.037) & \bf{(0.000)} && (0.218) & (0.038) & \bf{(0.019)} \\[2pt] 
& AC$_{2}$  & 10.138 & 10.547 & 10.635 & \bf{10.724} && 0.169 & \bf{0.008} & 0.020 \\[2pt] 
&   & (-0.571) & (-0.161) & (-0.074) & \bf{(0.016)} && (0.595) & (0.074) & \bf{(0.025)} \\[2pt] 
& AC$_{3}$  & 10.389 & 9.813 & 10.745 & \bf{10.693} && 0.178 & \bf{0.008} & 0.018\\[2pt]
&   &  (-0.320) & (-0.896) & (0.037) & \bf{(-0.016)} && (0.366) & (0.038) & \bf{(0.024)}  \\ \\[-2pt]
{\it (2)}&\multicolumn{9}{l}{{\it Pareto + $t$-copula: \changed{True AC = ( 7.198, 8.908, 12.206)} } }\\[5pt]
& AC$_{1}$  & 6.835 & 8.162 & 7.697 & \bf{7.339} && 0.238 & \bf{0.010} & 0.041 \\[2pt]
&   &  (-0.362) & (0.964) & (0.499) &  \bf{(-0.121)} && (0.433) & (0.499) & \bf{(0.132)} \\[2pt] 
& AC$_{2}$ & 8.785 & 8.355 & 8.740 & \bf{8.765} && 0.223 & \bf{0.010} & 0.028\\[2pt] 
&   &   (-0.122) & (-0.553) & (-0.167) &  \bf{(-0.023)} && (0.255) & (0.168) & \bf{(0.046)} \\[2pt] 
& AC$_{3}$ & 11.913 & 11.781 & 11.875 & \bf{12.208} && 0.134 & \bf{0.006} & 0.024 \\[2pt] 
& & (-0.293) & (-0.426) & (-0.332) & \bf{(0.144)} && (0.322) & (0.332) & \bf{(0.148)} \\ \\[-2pt]
{\it (3)}&\multicolumn{9}{l}{{\it Student's $t$ + survival Clayton: True AC =  (5.647, 5.647, 5.647)} }\\[5pt]
& AC$_{1}$ & 5.592 & 5.693 & \bf{5.662} & 5.617 && 0.081 & \bf{0.006} & 0.018  \\[2pt] 
&   & (-0.055) & (0.046) & \bf{(0.015)} & (-0.029) && (0.098) & \bf{(0.016)} & (0.034) \\[2pt] 
&AC$_{2}$ & 5.410 & 5.722 & \bf{5.642} & 5.665 && 0.079 & \bf{0.006} & 0.019 \\[2pt] 
&   &  (-0.236) & (0.076) & \bf{(-0.005)} & (0.018) && (0.249) &\bf{(0.007)} & (0.026) \\[2pt] 
&AC$_{3}$ & 5.473 & 5.517 & \bf{5.636} & 5.658 && 0.082 & \bf{0.006} & 0.018 \\[2pt] 
&   &  (-0.173) & (-0.130) & \bf{(-0.011)} & \bf{(0.011)} && (0.192) & \bf{(0.012)} & (0.021) \\ \\[-2pt]
{\it (4)}&\multicolumn{9}{l}{{\it  Student's $t$ + $t$-copula: True AC =  (2.996, 3.745, 6.741)} }\\[5pt]
& AC$_{1}$ & 2.821 & 3.065 & \bf{2.997} & 2.940 && 0.117 & \bf{0.007} & 0.036 \\[2pt] 
&   & (-0.176) & (0.069) & \bf{(0.001)} & (-0.056) && (0.211) & \bf{(0.007)} & (0.067) \\[2pt] 
& AC$_{2}$ & 3.772 & 3.560 & \bf{3.742} & 3.792 && 0.109 & \bf{0.006} & 0.033  \\[2pt] 
&   &  (0.027) & (-0.185) & \bf{(-0.004)} & (0.047) && (0.112) & \bf{(0.007)} & (0.057) \\[2pt] 
& AC$_{3}$ & 6.564 & 6.852 & \bf{6.745} & 6.751 && 0.043 & \bf{0.002} & 0.011 \\[2pt] 
&   &  (-0.178) & (0.110) & \bf{(0.003)} & (0.010) && (0.183) & \bf{(0.004)} & (0.015)  \\ \\[-2pt]
   \bottomrule[1.5pt]
\end{tabular}
\begin{tablenotes}
\item[$\dagger$] {\footnotesize The estimate is computed for the Monte Carlo $\bm{MC}$, Nadaraya-Watson $\bm{NW}$, generalized regression $\bm{GR}$, and Metropolis-Hastings $\bm{MH}$ estimators.
The standard error is computed except for the NW estimator.
The sample size is $N=10^6$ for all methods.
}
\end{tablenotes}
\end{threeparttable}
\end{table}

The results of estimation are summarized in Table $\ref{Table 1}$.
In the four different risk models (1)--(4), we report the estimates, their approximated standard errors, biases, and rooted MSEs (RMSEs) of the four different estimators: MC, NW, GR, and MH.

In the first risk model, true VaR contributions are obtained by equally allocating the total VaR since the marginal distributions are homogeneous and the copula is exchangeable.
We observed that the MC and NW estimators have relatively larger biases than those of others.
Compared with the MH estimator, the GR estimator still suffers from some inevitable bias although its standard error is quite small.
The MC estimator has a relatively large standard error due to sample inefficiency.
Overall, the MH estimator outperforms all other estimators in terms of RMSE.

The second risk model does not allow us to analytically calculate the true VaR contributions. 
Therefore, the true VaR contributions are computed by Monte Carlo numerical integration, which still works with enough accuracy for dimension three.
We can observe that existing estimators suffer from biases possibly caused by asymmetry and multi-modality of the conditional distribution $F_{\bX'|S=v}$.
In particular, the GR estimator has relatively large bias and RMSE in contrast to the good performance in the first risk model.
On the other hand, the MH estimator maintains lower bias and RMSE compared with the other estimators.

In the third risk model, the true VaR contributions are given by the equal allocation based on the same discussion as in case (1).
Thanks to the symmetry and unimodality of the conditional distribution $F_{\bX'|S=v}$, all estimators retain small biases and RMSEs.
Together with the results in cases (1) and (2), one can state that the GR estimator performs well so long as $F_{\bX'|S=v}$ is symmetric and unimodal.
Additionally, the MH estimator reduces bias and RMSE compared with those of MC and NW estimators.

The final risk model provides the true VaR contributions via the formula $(\ref{VaR contributions elliptical case})$.
In such an elliptical case, the GR estimator provides quite an accurate estimate.
Although the conditional distribution $F_{\bX'|S=v}$ is heavy-tailed as seen in  Figure $\ref{QF_Figure1}$ (iv), the MH estimator retains high performance compared with the MC and NW estimators.
The bias of the MH estimator is significantly improved compared with the MC and NW estimators.
Moreover, its standard error and RMSE are lower than those of the MC estimator. 

Throughout the numerical study, the MH estimator provided small bias and RMSE regardless of the shape of the conditional distribution $F_{\bX'|S=v}$.
In the case when $F_{\bX'|S=v}$ is unimodal and symmetric, the GR estimator also guarantees a good performance.
On the other hand, at least in our numerical experiment, the MC and NW estimators had relatively larger biases and RMSEs compared with other estimators.

\subsection{Empirical study}\label{subsec: empirical study}

The numerical study is now extended to a high-dimensional case with real-world data.
We used the dataset \textsf{stockdata} in \textsf{R}-package \textsf{huge}, which consists of stock market data of closing prices from all stocks in the S\&P 500 for all the days the market was open in the period of January 1, 2003 to January 1, 2008 (five years).
During the time period, there remained $d=452$ stocks in the S\&P 500.
The sample size is $T=1258$. 
We transformed the data into the log-ratio of the price at time $t$ to the price at time $t-1$.

Most stylized facts on stock returns listed in Chapter 3 of \citet{mcneil2015quantitative} are observable in the data.
For example, return series are unimodal, leptokurtic, and heavy-tailed with little serial correlation and volatility clusters.
Moreover, the $d$ return series are mutually dependent.
Taking these observations into account, we adopted a copula-GARCH model with skew-$t$ white noise \citep[ST-GARCH; see, for example,][]{jondeau2006copula,huang2009estimating}.
In the model, $d$ marginal time series are modelled by GARCH$(1,1)$ and the underlying white noise processes follow skew-$t$ distributions with an inhomogeneous degree of freedom $\nu_j>0$ and skewness parameter $\gamma_j>0$; that is, within a fixed time period $\{1,\dots,T\}$ the $j$-th return series $(X_{1,j},\dots,X_{T,j})$ follows 
\begin{align*}
X_{t,j}=\mu_j + \sigma_{t,j}Z_{t,j},\quad \sigma_{t,j}^2 = \omega_{j}+\alpha_{j}X_{t-1,j}^2+\beta_{j}\sigma_{t-1,j}^2,\quad Z_{t,j}\iidsim \text{ST}(\nu_j,\gamma_j)
\end{align*}
for $t=2,\dots,T,\quad  j=1,\dots,d$, where $\omega_{j} >0,\alpha_{j},\beta_{j}\geq 0$, $\alpha_{j}+\beta_{j}<1$, and $Z_{t,j}$ follows a skew-$t$ distribution $\text{ST}(\nu_j,\gamma_j)$ with density given by
\begin{align}\label{density of skew t distribution}
f_{j}(x_j; \nu_j,\gamma_j)=\frac{2}{\gamma+\frac{1}{\gamma}}\left\{t(x_j,\nu_j)1_{[x_j\geq 0]} + t(\gamma_j x_j,\nu_j)1_{[x_j< 0]}\right\}
\end{align}
where $t(x,\nu)$ is a probability density function of a student's $t$-distribution with degree of freedom $\nu>0$ and a skewness parameter $\gamma >0$ with $\gamma=1$ symmetric; see \citet{fernandez1998bayesian} for more detail. 
The copula among $\bZ_t = (Z_{t,1},\dots,Z_{t,d})$ is assumed to be a student's $t$-copula with parameters $\nu$ and ${\bf P}$ independent of time $t$.

We estimated parameters of the ST-GARCH(1,1) model with the $t$-copula based on the copula approach.
First, we fitted the ST-GARCH(1,1) model with the maximum likelihood method to the marginal time series.
Then, to obtain pseudo-samples from the copula of $\bZ$, distributional transform was applied to the $d$-dimensional white noise process extracted from the ST-GARCH model.
We finally fit the $t$-copula to them with method-of-moments using Kendall's tau for the dispersion matrix ${\bf P}$ and the maximum likelihood method for the degree of freedom $\nu$; see \citet{demarta2005t} for more detail.
The results of the estimation are summarized in Figure~\ref{Boxplots:parameters} with each boxplot representing $d$ numbers of each parameter.
From (B1) and (B5), the estimates of means and omegas are almost $0$. 
From (B3), most of the marginal white noise distributions are symmetrical but some are skewed.
From (B4), their degrees of freedom range from two to ten, that is, the tail-heaviness of the return series is inhomogeneous over $d$ assets.
Finally, (B8) shows that the pairwise correlations among the return series are typically from 
0.2 to 0.4, and some have strong positive correlations.

\begin{figure}[t]
\center
\includegraphics[scale=0.6]{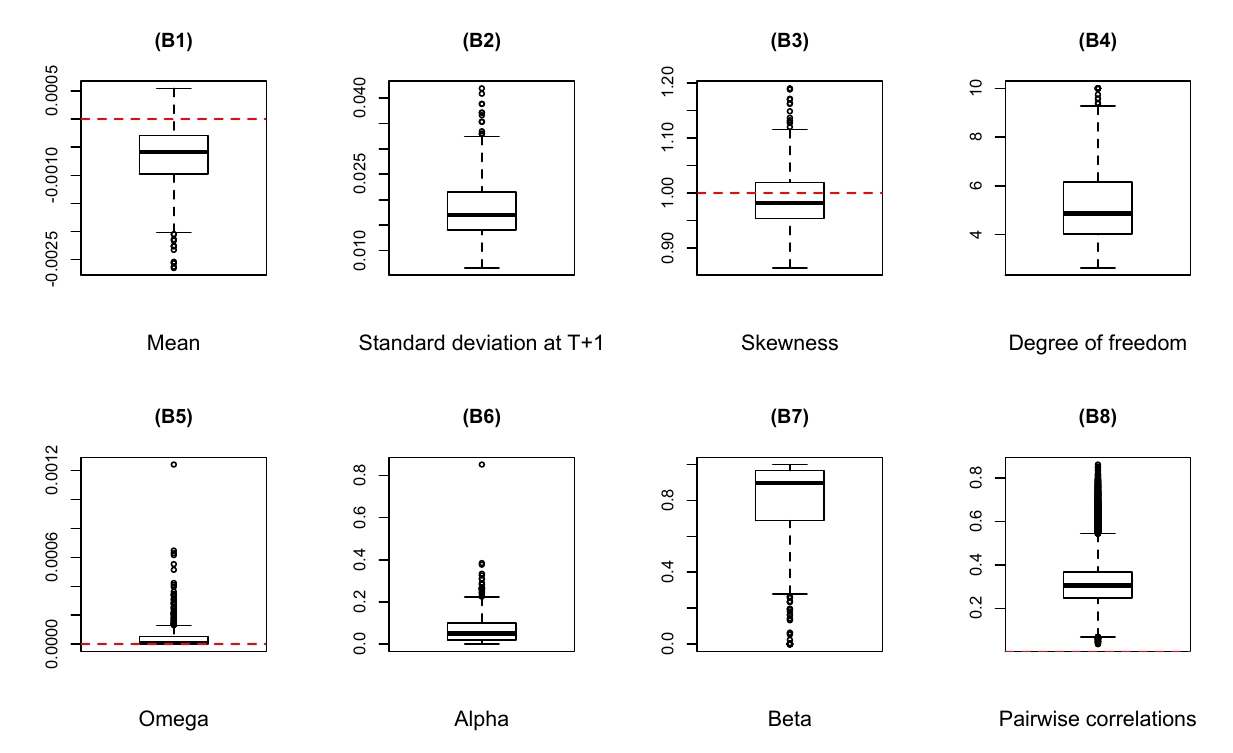}
\caption{Boxplots of $d$ estimates of each parameter 
(B1) $\mu_j$, (B2) $\sigma_{T+1,j}$, (B3) $\gamma_j$, (B4) $\nu_j$, (B5) $\omega_j$, (B6) $\alpha_j$, (B7) $\beta_j$, and (B8) $\rho_{j_1,j_2}$ for $j=1,\dots,d$ and $j_1,j_2 \in \{1,\dots,d\}$ of ST-GARCH(1,1) model
$X_{t,j}=\mu_j + \sigma_{t,j}Z_{t,j},\ \sigma_{t,j}^2 = \omega_{j}+\alpha_{j}X_{t-1,j}^2+\beta_{j}\sigma_{t-1,j}^2,\ Z_{t,j}\iidsim \text{ST}(\nu_j,\gamma_j)$, $t=1,\dots,T+1$, $j=1,\dots,d$ with a $t$-copula with parameters $\nu$ and ${\bf P}$. 
The estimate of the degree of freedom of the $t$-copula was $\nu=89.039$.}
\label{Boxplots:parameters}
\end{figure}

Our goal in this study is to compute the conditional VaR contributions at time $T+1$ given the history $\mathcal F_{t}$.
Under the model described above, the marginal distribution of the $j$-th return at time $T+1$ is $X_{T+1,j|\mathcal F_{T}}\sim$ST$(\mu_j,\sigma_{t+1,j}^2,\nu_j,\gamma_j)$, where ST$(\mu_j,\sigma_{t+1,j}^2,\nu_j,\gamma_j)$ is a skew $t$-distribution with density $f_{j}(\frac{x_j-\mu_j}{\sigma_{t+1,j}}; \nu_j,\gamma_j)$ with $f_{j}(\cdot;\nu_j,\gamma_j)$ defined in \eqref{density of skew t distribution}.
Their copula is a student's $t$-copula with parameters $\nu$ and $P$.
Based on this multivariate model, conditional VaR contributions at time $T+1$ given histories $\mathcal F_t$ are estimated by the same procedure as in section~\ref{subsec: simulation study}.

We estimated the conditional VaR contributions (AC$_1^{T+1},\dots,$AC$_d^{T+1}$) with confidence level $p=0.999$ by using MC, NW, GR, and MH methods.
In MC, $N=10^5$ samples were generated and the total VaR was estimated as the $Np$-th largest sample among them.
The run time of the MC simulation was $2.690$ minutes.
The MC estimates of VaR contributions were then computed as sample means of the conditional samples whose sums fall into the set $A_{\delta}=[v-\delta,v+\delta]$.
The bandwidth was set to be $\delta = 4.8$ so that there were $M_{\delta,N} = 733$ conditional MC samples.
Estimates of standard errors were also computed based on these samples.
NW, GR, and MH estimators were computed analogously to the previous simulation study in section~\ref{subsec: simulation study}.
For the MH estimator, the MpCN proposal distribution was chosen since the target distribution was expected to be heavy-tailed and elliptical to some extent.
The length of the sample path was chosen to be $N=10^4$, and the run time of the MH algorithm was $5.487$ minutes.
We inspected the autocorrelation plots and ACR to check the validity of the proposal distribution.
We observed that all autocorrelations decreased below $0.1$ if lags are larger than $40$.
Together with the ACR $0.983$, we concluded that the choice of $q$ was appropriate.

Figure~\ref{estimatedACsPlots} shows the MC, NW, GR, and the MH estimates of the conditional VaR contributions (AC$_1^{T+1},\dots,$ AC$_{d}^{T+1}$) of returns at time $T+1$ given histories $\mathcal F_{T}$ plotted with the homogeneously allocated capitals VaR$_{p}(S|\mathcal F_{T})/d$ and the standardized marginal VaRs defined by $\text{VaR}_{p}(X_{T+1,j}|\mathcal F_{T})\Delta_{p}(\bX_{T+1}|\mathcal F_{T})$, where $\Delta_{p}(\bX_{T+1}|\mathcal F_{T})$ is the so-called \emph{superadditivity ratio} defined by
\begin{align*}\label{eq:standardized VaR}
\Delta_{p}(\bX_{T+1}|\mathcal F_{T})=\frac{\text{VaR}_{p}(S|\mathcal F_{T})}{\sum_{j=1}^d \text{VaR}_{p}(X_{T+1,j}|\mathcal F_{T})}.
\end{align*}
For MC, GR, and MH estimators, the 95\% confidence upper and lower bounds are also plotted.
On the x-axis, the 452 assets are rearranged in increasing order of MH estimators.

\begin{figure}[t]
\center
\includegraphics[scale=0.6]{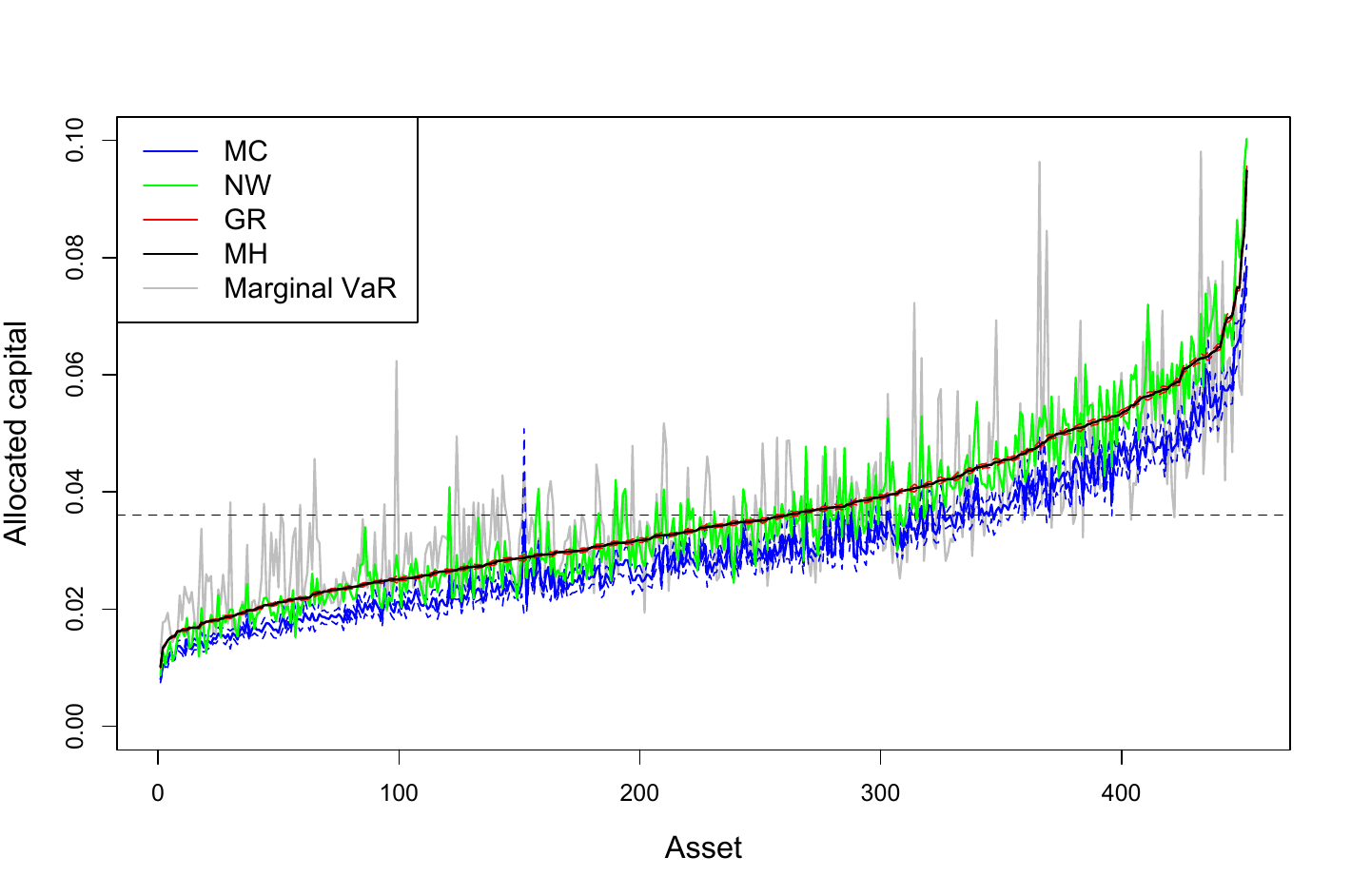}
\caption{
Monte Carlo (MC; blue), Nadaraya-Watson (NW; green), Generalized regression (GR; red), and Metropolis-Hastings (MH; black) estimators of conditional VaR contributions at time $T+1$ given $\mathcal F_T$ plotted with standardized marginal value-at-risks (gray) and the homogeneously allocated capitals (dotted black line).
The colored dotted lines represent the 95\% confidence upper or lower bounds of MC, GR, or MH estimators.
}
\label{estimatedACsPlots}
\end{figure}

Compared with the dotted line representing the homogeneous allocation, all the estimated allocated capitals show inhomogeneity among assets.
Overall, the estimated VaR contributions are less volatile than the standardized marginal VaRs, which implies the benefit of the diversification effect.
We can also observe that the MH estimates and GR estimates almost coincide for all $d$ assets.
The confidence intervals of both estimators are much tighter than that of the MC estimator.
NW estimates fluctuate around the line of the MH and GR estimates.
On the other hand, the MC estimates deviate from these lines, which indicates that the MC estimators contain inevitable biases.
In summary, although the true ACs are unknown, the GR and MH estimators retain stable performance compared with the MC and NW estimators even if the dimension $d$ is large and marginal distributions are inhomogeneous.

\subsection{Advantages and disadvantages of the MH estimator}
We summarize the advantages and disadvantages of the MH estimator compared with the other estimators.
The first advantage is that the MH estimator is consistent whereas this is not always true for the other estimators.
As explained in section $\ref{capital allocation problem}$, the MC, NW, and GR estimators have biases which cannot be easily eliminated.
In fact, we observed in Table $\ref{Table 1}$ that unignorable biases of the MC, NW, and GR estimators sometimes remain even when their standard errors are sufficiently small.
In contrast, the MH estimator provides more accurate estimates of VaR contributions as $N \rightarrow \infty$ due to its consistency.
Since CLT also holds, the confidence interval of the true VaR contributions is also available.
Secondly, the MH estimator has great sample efficiency compared with the MC estimator.
While samples are generated from $F_{\bX}$ and most are discarded in the MC method, no samples are wasted in the MH method since it directly simulates $F_{\bX|S=v}$.
Consequently, the MH estimator can achieve low standard errors.
Finally, the MH estimator can maintain high performance even when the conditional distribution $F_{\bX'|S=v}$ is multimodal or heavy-tailed.
As discussed in subsection~\ref{results and discussions}, the performance of the GR estimator highly depends on the shape of $F_{\bX'|S=v}$.
On the other hand, for the MH estimator, the shape of $F_{\bX'|S=v}$ can be directly captured through the proposal distribution $q$.
By choosing an appropriate proposal distribution $q$ according to the shape of $F_{\bX'|S=v}$, the MH estimator can attain great performance.
This advantage, however, can be seen as a disadvantage from the viewpoint of the simplicity of estimation.
In general, estimation with MH requires two steps: first is to choose a family of proposal distribution, and the second is to determine its parameters.
The second step of parameter estimation can be based on the MC samples falling into set $A_{\delta}$, which is regarded as the {\it pseudo} samples from $F_{\bX|S=v}$.
Meanwhile, the first step is not so straightforward.
We will discuss this issue in the next subsection $\ref{Empirical study of proposal distributions}$.
Another disadvantage of the MH estimator is that it typically requires a longer run time than other existing estimators. 
Since MH requires $N$ times of simulating the proposal distribution and evaluating the acceptance probability $\eqref{acceptance probability}$, careful programming and proposal selection are necessary to save computational time.

\subsection{Guidelines for the choice of proposal distribution}\label{Empirical study of proposal distributions}
A significant drawback of the MH estimator is that the choice of an appropriate proposal distribution $q$ is not as simple as the parameter selections of other existing estimators. 
An instruction for proposal selection is necessary since it highly affects the performance of the MH.
In this subsection, we first investigate the symptoms caused by an inappropriate choice of $q$.
Then, we consider how to overcome these problems based on the numerical experiments provided above.
Practical guidelines for choosing an appropriate proposal distribution are also provided.

An inappropriate choice of $q$ is largely classified into two cases.
One is that proposal distribution $q$ often generates a candidate of which the probability measured by $\pi$ is quite small.
This case occurs, for example, when $q$ does not fully capture the shape of $\pi$.
In such a case, the Markov chain moves quite slowly and this yields a high asymptotic standard error of the MH estimator.
This symptom appears as quite a low acceptance rate and high autocorrelations.
Another case is wherein $q$ generates only some parts of the whole support of $\pi$.
This case occurs, for example, when $\pi$ has distinct local modes and the variance of $q$ is so small that the chain cannot pass between ridges.
In such a case, an estimate can be significantly biased, although the acceptance rate and autocorrelation plots are seemingly perfect.
This symptom appears as a distorted plot of MCMC samples whose shapes are completely different from the target distribution $\pi$.

How can we detect and avoid such fallacious estimates?
First, as mentioned in section $\ref{choice of the proposal distribution}$, it is indispensable to inspect the autocorrelation plots and ACR to prevent the first symptom.
Additionally, to avoid the second symptom, we recommend drawing the plots of the generated Markov chain and comparing them with the plots of the MC samples whose componentwise sums belong to $A_{\delta}=[v-\delta,v+\delta]$.
Since such MC samples follow the distribution $F_{\bX|S\in A_{\delta}}$, one can detect the distortion of the generated Markov chain by comparing the two scatter plots of $F_{\bX|S=v}$ and $F_{\bX|S\in A_{\delta}}$.
As an example from our simulation study in subsection~\ref{subsec: simulation study}, Figure $\ref{QF_Figure2}$ shows the scatter plots of the MC samples whose sums belong to $[v-\delta,v+\delta]$ overlaid on the scatter plots of the MH samples.
In the figure, we can check that the shapes of the scatter plots of the MH samples bear striking resemblance to those of the MC samples for all risk models.
If some part of the support of $\pi$ is covered by the MC samples but not by the MH samples, the choice of $q$ is questionable.

\begin{figure}[t]
\center
\vspace{-10mm}
\includegraphics[scale=0.5]{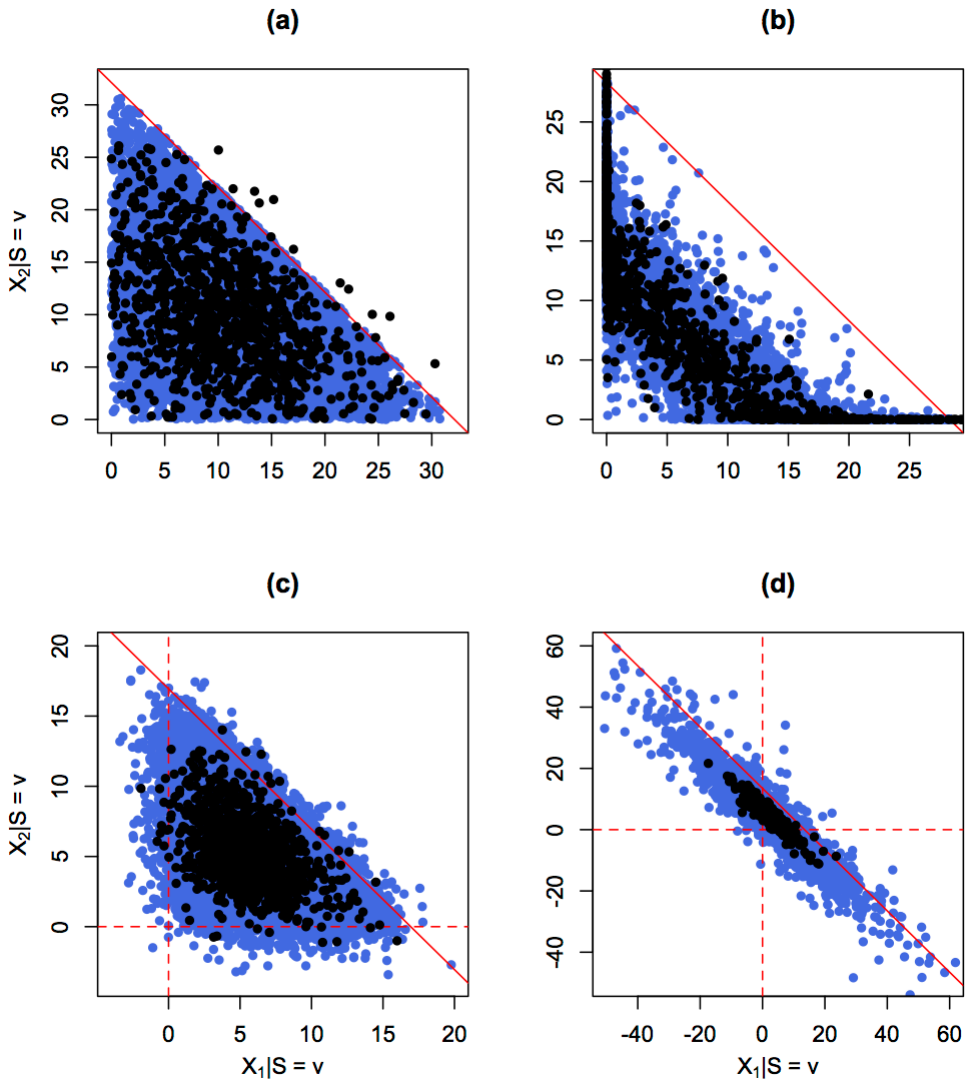}
\vspace{-25mm}
\caption{Scatter plots of the Monte Carlo (MC; black) and Metropolis-Hastings (MH; blue) samples for different risk models: (a) Pareto + survival Clayton, (b) Pareto + $t$-copula, (c) Student's $t$ + survival Clayton, (d) Student's $t$ + $t$-copula.
The red lines represent the edges of the $v$-simplex, where $v$ is the estimate of VaR$_{p}(S)$.
We plot the MC samples generated from $F_{\bX}$ such that their sums belong to $A_{\delta}=[v-\delta,v+\delta]$.
In the four risk models, the values of $\delta$ are (1) 4.8, (2) 3.9, (3) 2.2, and (4) 1.7.
When drawing the scatter plots of the MH samples, we used subsamples that are picked up every 100th point among the original Markov chains.
}
\label{QF_Figure2}
\end{figure}

Finally, through numerical experiments we found that dependence information of the underlying risk model can be helpful for the selection of $q$.
When copula $C$ of the underlying risk model only has positive dependences for all pairs of loss variables, then the conditional distribution $F_{\bX|S=v}$ is likely to be unimodal and light-tailed since positive dependence among $\sqc{X}{d}$ prevents them from being diversified under the constraint $\{X_{1}+\cdots+X_d = v\}$.
In risk models (1) and (3) in subsection~\ref{description of the numerical comparison}, where copula $C$ has only positive dependences, the contour plots in Figures $\ref{QF_Figure1}$ (i) and (iii) show that $F_{\bX'|S=v}$ is unimodal and light-tailed.
These features facilitate the estimation with MH since simple proposal distributions such as the random walk proposal ($\ref{random walk proposal}$) and independent proposal  $(\ref{independent proposal})$ can perform well.
Conversely, when copula $C$ has negative dependence, $F_{\bX'|S=v}$ tends to be multimodal or heavy-tailed since negative dependence allows each component of $\bX$ to take extreme values under $\{X_1+\cdots+X_d = v\}$.
In risk models (2) and (4) in subsection~\ref{description of the numerical comparison}, where copula $C$ has negative dependences, Figure $\ref{QF_Figure1}$ (ii) indicates that $F_{\bX'|S=v}$  is bimodal, and the contour plot in Figure $\ref{QF_Figure2}$ (d) shows that $F_{\bX'|S=v}$ is heavy-tailed.
In such cases, careful proposal selection is required for achieving an efficient MH estimator. 
When the losses $\sqc{X}{d}$ are all nonnegative, then $F_{\bX'|S=v}$ is supported on the bounded simplex $\mathcal S_{v}$ defined in \eqref{v simplex}.
Therefore, one can cover the whole support of $F_{\bX'|S=v}$ by choosing $q$ as the independent proposal with the distribution defined on the simplex.
Uniform distribution on $\mathcal S_{v}$ can be the safest choice.
It is also possible to choose other distributions that share the same features of $F_{\bX'|S=v}$ observed in the MC samples. 
For instance, since bimodality is observed in the contour plot in Figure $\ref{QF_Figure2}$ (b), we choose $q$ as the independent proposal distribution with $f$ the Dirichlet distribution on $\mathcal S_{v}$, which can possess two distinct modes around the edges of the simplex.
When $\bX$ is $\R^{d}$-valued and negatively dependent, an efficient MCMC is challenging since the target distribution $F_{\bX'|S=v}$ is likely to be multimodal or heavy-tailed.
As a special case, when $F_{\bX}$ is elliptical to some extent, then $F_{\bX'|S=v}$ is likely to be elliptical again.
In such a case, even if it is heavy-tailed, the MpCN proposal distribution $(\ref{MpCN proposal})$ is known to perform well, which is also demonstrated by the simulation study of the risk model (4) in subsection~\ref{description of the numerical comparison} and by the empirical study in subsection~\ref{subsec: empirical study}.

The discussions on choosing an appropriate proposal distribution are summarized as a flowchart in Figure $\ref{flowchart of the choice of the proposals}$.
Together with the guidelines, the whole procedure of our MH estimator of VaR contributions presented in this paper is summarized as follows.

\begin{flushleft}
{\bf Algorithm 3:} (Estimation of VaR contributions with MCMC)
\end{flushleft}
\begin{enumerate}\setlength{\itemsep}{-5pt}
 \item[1.] Generate $\sqc{\bX}{M} \iidsim F_{\bX}$ by MC.\\
 \item[2.] Based on the samples generated in step 1, estimate VaR by $v=\widehat{\text{VaR}}_{p}(S)$.\\\item[3.] For a bandwidth $\delta>0$, extract subsamples such that $\tra{\bone}_{d}\bX_{m} \in [v-\delta,v+\delta]$ for $m=1,\dots,M$.\\
 \item[4.] Choose a family of proposal distributions according to the guideline in Figure $\ref{flowchart of the choice of the proposals}$.\\
 \item[5.] Based on the pseudo samples extracted in step 3, determine the parameters of the proposal distribution $q$.\\
 \item[6.] For a sample size $N>0$, proposal density $q$ and the initial value $\bm X^{(0)}=\bx^{(0)}$, run Algorithm 1 to generate an $N$-path $(\bX^{(1)},\dots,\bX^{(N)})$ of a Markov chain whose stationary distribution is $f_{\bX|S=v}$.\\
 \item[7.] To check the validity of proposal distribution $q$, compute the acceptance rate, draw the autocorrelation plots, and compare the scatter plots of the MC and MH samples.\\
  \item[8.] If the proposal selection is verified in step 7, set the MH estimator of VaR contributions $\eqref{MCMC estimator}$ based on the sample path generated in step 6. Otherwise, go to step 4 and choose another proposal distribution.\\
\end{enumerate}

\begin{figure}[h]
\begin{center}
\begin{tikzpicture}[
    node distance=2cm,
    startstop/.style={rectangle, draw=black, thick, fill=white,
    text width = 5em, text centered, minimum height= 2em},
    process/.style={rectangle, draw=black, thick, fill=white,
    text width = 15em, minimum height= 2em},
    ]
	\node (node0)[startstop]
		{Start};
    \node (node1) [process, below of = node0,yshift=-1cm]
    	{Are the loss random variables $(\sqc{X}{d})$ {\bf positively dependent} on each other?};
    	\draw [arrows=-Stealth] (node0) -- (node1);
    \node (node2) [process, right of=node1, xshift=+7cm]   
    	{The conditional distribution $F_{\bX|S=v}$ is likely to be unimodal and light-tailed. Thus, simple proposal distributions such as the {\bf random work} and {\bf independent proposals} can work well.};
        \draw [arrows=-Stealth] (node1) --node[midway,yshift=+0.2cm] {Yes} (node2);
 	\node (node3) [process, below of=node1,yshift=-3cm]              
 		{Are the losses all {\bf nonnegative}, that is, $\sqc{X}{d}\geq 0$\\
 		(or, are their negative parts bounded, that is, there exist \\
 		$c_j=\text{ess.inf}(X_{j}) > -\infty$ such that \\
 		$X_j-c_j \geq 0$ for $j=1,\dots,d$)?};
 	    \draw [arrows=-Stealth] (node1) --node[midway,xshift=-0.3cm] {No} (node3);
 	\node (node4)[process, right of = node3, xshift=+7cm]
 	 	{Although $F_{\bX|S=v}$ could be multimodal, it is supported on the bounded simplex. 
 	 	Therefore, {\bf independent proposal distribution defined on the simplex} could be a sensible choice. Uniform distribution on the simplex is the safest choice, but other distributions such as Dirichlet distribution are also reasonable.};
 	 	 \draw [arrows=-Stealth] (node3) --node[midway,yshift=+0.2cm] {Yes} (node4);
 	\node (node5)[process, below of=node3,yshift=-3cm]{Is the joint distribution of the loss random vector $\bX$ {\bf elliptical}?};
 	 	 \draw [arrows=-Stealth] (node3) --node[midway,xshift=-0.3cm] {No} (node5);
 	\node (node6)[process, right of = node5, xshift=+7cm]
 	 	{Although $F_{\bX|S=v}$ can be heavy-tailed, it is likely to be elliptical. Therefore, the {\bf MpCN proposal distribution} can work well.};
 	 	 \draw [arrows=-Stealth] (node5) --node[midway,yshift=+0.2cm] {Yes} (node6);
 	\node (node7)[process, below of=node5,yshift=-2cm]
 		{No guideline is available.};
 		\draw [arrows=-Stealth] (node5) --node[midway,xshift=-0.3cm] {No} (node7);
  \end{tikzpicture}
 \end{center}
 \caption{Flowchart for choosing the proposal distribution of the Metropolis-Hastings (MH) estimator of value-at-risk contributions.}\label{flowchart of the choice of the proposals}
\end{figure}
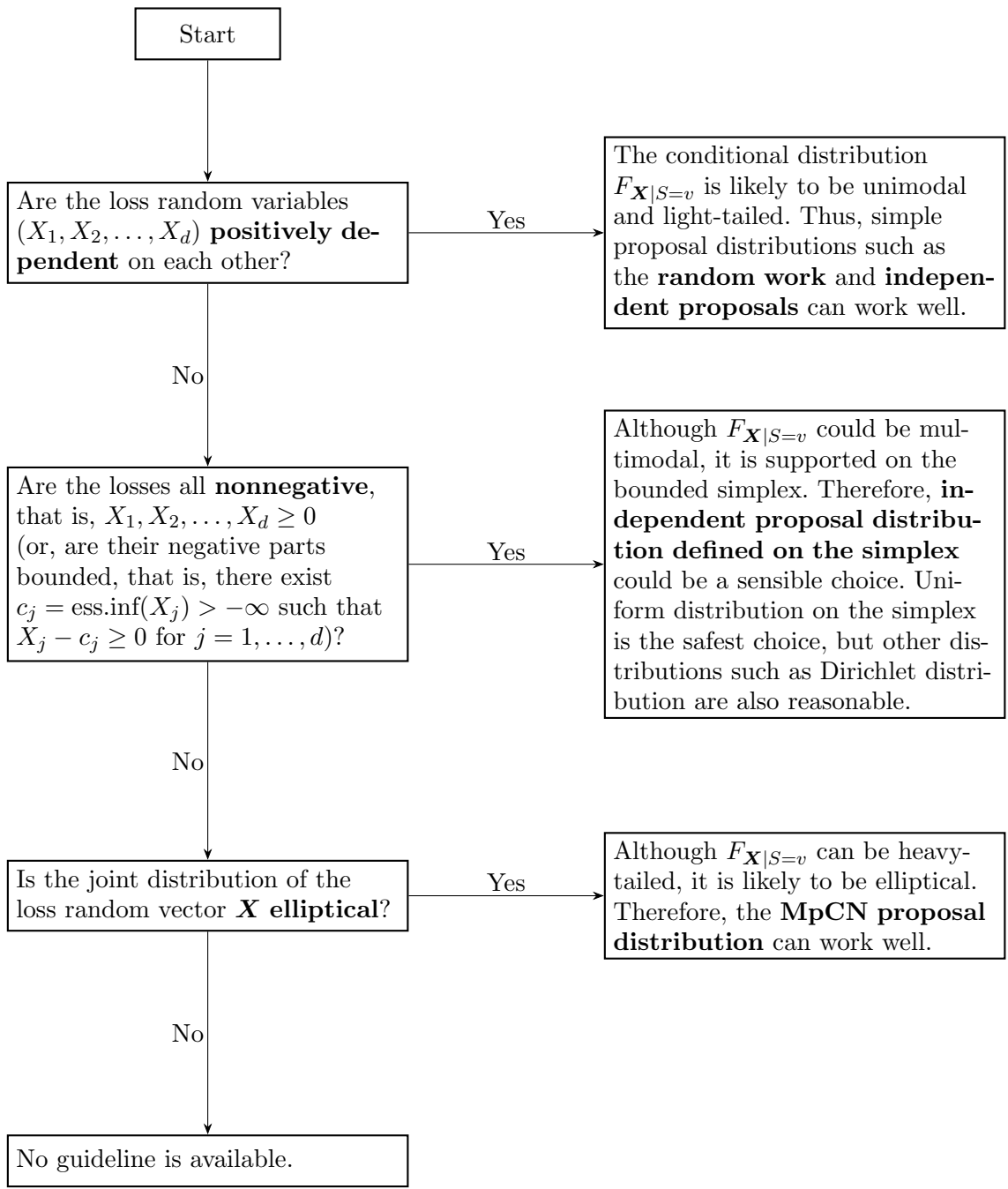

\section{Concluding remarks}\label{Concluding remarks}
Computing VaR contributions for a risk model specified by joint density is generally a difficult task.
To achieve this, the MH estimator of VaR contributions is proposed.
Its sample efficiency is significantly improved since the MH method generates samples directly from the conditional density given the sum constraint.
Moreover, since the MH estimator can capture the features of the risk model more directly than the existing estimators, it can maintain high performance even when the underlying loss distribution is multimodal or heavy-tailed.
By the general theory of Markov chains, the MH estimator is consistent and asymptotically normally distributed.
Through simulation and empirical studies based on real-world data, the performance of the MH estimator was compared with those of other existing estimators for various risk models.
The numerical results demonstrated that in most risk models, the MH estimator had smaller bias and RMSE compared with other existing estimators even when the dimension of the portfolio was high, such as $d\approx 500$.

Potential future research includes a theoretical study of the conditional joint distribution of $\bX|S=v$.
Our main interest is in the influence of the underlying copula of a risk model on the tail behavior and multimodality of the density $f_{\bX|S=v}$.
We believe that revealing relationships among them can provide more promising guidelines for the  proposal selection of the MH estimator.

\section*{Acknowledgements}
We wish to thank to Paul Embrechts from ETH Z\"urich for his valuable comments regarding the simulation setup.
We would also like to express our gratitude to Kengo Kamatani from Osaka University, and Marius Hofert from the University of Waterloo for fruitful discussions on MCMC and Archimedean copulas.

\section*{Funding}
This work was supported by the Japan Society for the Promotion of Science (JSPS) under the Core-to-Core program at Keio University.



\bibliographystyle{rQUF}
\bibliography{VaRcont}

\appendices

\section{Consistency and asymptotic normality}\label{sec: consistency and asymptotic normality}

In this appendix, we derive conditions on a copula and marginal distributions with which the corresponding MH estimator of VaR contributions satisfies consistency~\eqref{MCMC estimator consistency} and CLT $(\ref{central limit theorem})$ for some choice of proposal distribution $q$.
This study reveals which proposal distribution is appropriate for a given risk model.

We classify the loss distribution $F_{\bX}$ into two cases; one wherein $\text{supp}(f_{\bX})=\R^{d}_{+}=\{\bx \in \R^{d}:\bx \geq \bzero\}$ and another wherein $\text{supp}(f_{\bX})=\R^{d}$. 
The former corresponds to the case wherein we model {\it pure losses}, and the latter to the case of \emph{profits and losses (P}\&\emph{L)}.
Our result is mainly about the former case, and we provide some limited examples for the latter case.
It should be emphasized that the former case of pure losses includes a broad range of loss models.
To demonstrate this, let $c_{j}=\text{ess.inf}(X_j)$, and set $\tilde X_{j}=X_j - c_j$, $j=1,\dots,d$. 
If $-\infty<c_j$, then $\tilde X_j \geq 0$.
For $\tilde S = \sum_{j=1}^{d}\tilde X_j$, the translation invariance of VaR$_{p}$ implies that
\begin{align*}
\text{VaR}_{p}(\tilde S) = \text{VaR}_{p}(S) - \sum_{j=1}^{d}c_{j}. 
\end{align*}
Therefore, the allocated capital of $\tilde X_j$ is given by
\begin{align*}
\tilde{\text{AC}}_{j} &= \E[\tilde X_j\ | \ \tilde S = \text{VaR}_{p}(\tilde S)]
= \E\left[X_j - c_j|S-\sum_{j=1}^{d}c_j = \text{VaR}_{p}(S)-\sum_{j=1}^{d}c_j\right]\\
&= \E[X_j|S = \text{VaR}_{p}(S)]- c_j = \text{AC}_j - c_j.
\end{align*}
Consequently, one can estimate $(\sqc{\text{AC}}{d})$ by first estimating $(\sqc{\tilde{\text{AC}}}{d})$ based on the joint distribution of $(\sqc{\tilde X}{d})$ such that $\text{supp}(f_{\tilde \bX})=\R^{d}_{+}$, and then subtracting $(\sqc{c}{d})$ from $\tilde{\text{AC}}$.
Therefore, our result for the former case includes the case of P\&L where the minimums of the profits are bounded.

\subsection{Case of pure losses}

When $\text{supp}(f_{\bX})=\R^{d}_{+}$, the conditional distribution $F_{\bX'|S=v}$ is supported on the following bounded set called the $v$-{\it simplex}:
\begin{equation}\label{v simplex}
{\mathcal S}_{v}:=\{\bx \in \R^{d'}: \bx \geq \bzero, 0\leq x_{1}+\cdots+x_{d'}\leq v\}.
\end{equation}
Thanks to the compactness of the support, we can state simple conditions on the marginal loss densities and copula density, which leads to consistency and CLT of the MH estimator.
\begin{theorem}\label{main theorem}
Suppose that the joint distribution $f_{\bX}$ is supported on $\R^{d}_{+}$ and has marginal densities $\sqc{f}{d}$ and a copula density $c$.
Then, $\sqrt{N}$-CLT holds for the MH estimator $(\ref{MCMC estimator})$ of VaR contributions if the following conditions $(C1)-(C3)$ hold:
\begin{itemize}\setlength{\itemsep}{-0mm}
\item[(C1)]  $\epsilon:=\inf_{\bx,\by \in \mathcal S_{v}} q(\bx,\by)>0$,\\
\item[(C2)]  $f_{j}(x)$ is positive and bounded above for any $x \in [0,v]$ for $j=1,2,\dots,d$, and\\
\item [(C3)] $c(\bu)$ is positive and bounded above for any $\bu \in F_{1}([0,v])\times\cdots \times F_{d}([0,v])$.\\ 
\end{itemize}
\end{theorem}

\begin{proof}
According to Theorem 23 in \citet{roberts2004general}, $\sqrt{N}$-CLT holds if the Markov chain is {\it uniformly ergodic} whenever $\E[||\bX'||^{2}|S=v]<\infty$.
Since $\sqc{X}{d}\geq 0$, the moment the condition is satisfied by the inequality
\begin{equation*}
\E[X_{i}X_{j}|S=v]\leq\E[(X_{1}+\cdots+X_{d})^{2}|S=v]=v^{2}<\infty
\end{equation*}
for any $i,j\in \{1,2,\dots,d\}$.
Thus, it suffices to show that the Markov chain is uniformly ergodic.
According to Theorem 1.3 in \citet{mengersen1996rates}, the Markov chain is uniformly ergodic if (and only if) the minorization condition \citep{rosenthal1995minorization} holds on the whole space ${\mathcal S}_{v}$; that is, there exists a positive integer $n$, a positive number $\delta>0$, and a probability measure $\nu$ such that 
\begin{equation}\label{minorization condition}
K^{n}(\bx,A)>\delta  \nu(A),
\end{equation}
for any $\bx \in {\mathcal S}_{v}$ and $A\in {\mathcal B}_{v}$, where ${\mathcal B}_{v}:={\mathcal B}(\R^{d'})\cap{\mathcal S}_{v}$.
Our target distribution can be written as
\begin{align*}
\pi(\bx) = \frac{f_{\bX}(\bx)}{f_{S}(v)}
= \frac{c(F_1(x_1),\dots,F_{d}(x_{d}))}{f_{S}(v)}f_1(x_1)\cdots f_{d} (x_{d}),
\end{align*}
where $(\sqc{x}{d-1})\in \mathcal S_v$ and $x_d = v-\tra{\bone}_{d}\bx$.
Thus, by conditions $(C2)$, $(C3)$, and that ${\mathcal S}_{v}\subset [0,v]^{d'}$, we have
\begin{equation}\label{bounded condition on pi}
l:=\inf_{\bx \in {\mathcal S}_{v}}\pi(\bx)>0,\hspace{6mm}
u:=\sup_{\bx \in {\mathcal S}_{v}}\pi(\bx)<\infty.
\end{equation}
Using $(\ref{bounded condition on pi})$ and condition $(C1)$, the minorization condition can be checked as follows.
For any $\bx \in {\mathcal S}_{v}$, define
\begin{equation*}
Q_{\bx}:=\left\{
\by \in {\mathcal S}_{v}: \frac{\pi(\by)}{\pi(\bx)} \frac{q(\by,\bx)}{q(\bx,\by)}<1
\right\}.
\end{equation*}
Then, for any $A \in {\mathcal B}_{v}$, we have
\begin{eqnarray*}
K(\bx,A)&=&\int_{A}\{q(\bx,\by)\alpha(\bx,\by) +r(\bx)\delta_{\bx}(\by)\}\text{d}\by\\
&\geq&\int_{Q_{\bx}}q(\bx,\by) \min \left[1, \frac{\pi(\by)}{\pi(\bx)} \frac{q(\by,\bx)}{q(\bx,\by)}\right]d\by\\
&&+\int_{A\backslash Q_{\bx}}q(\bx,\by) \min \left[1, \frac{\pi(\by)}{\pi(\bx)} \frac{q(\by,\bx)}{q(\bx,\by)}\right]d\by\\
&=&\int_{Q_{\bx}}\frac{\pi(\by)}{\pi(\bx)} q(\by,\bx)d\by
+\int_{A\backslash Q_{\bx}}q(\bx,\by)d\by\\
&\geq&\frac{\epsilon}{u}\int_{Q_{\bx}} \pi(\by)d\by+\epsilon\int_{A\backslash Q_{\bx}}\frac{\pi(\by)}{u}d\by\\
&=&\frac{\epsilon}{u} \pi(A).
\end{eqnarray*}
Therefore, the minorization condition holds for $n=1$, $\delta=\frac{\epsilon}{u}>0$, and $\nu= \pi$.
Consequently, the Markov chain is uniformly ergodic, and thus $\sqrt{N}$-CLT holds.
Since the minorization condition \eqref{minorization condition} holds, consistency of $\hat{\bm \pi}_{N}(\bm h)$ follows by Theorem 1 in \citet{nummelin2002mc}.
\end{proof}

An example of the pair of risk model and proposal distribution is given in the following example.

\begin{example}\label{example: main theorem}
For $j=1,\dots,d$, let $X_j$ follow Pareto distribution with density given by
\begin{equation}\label{pareto density}
f_j(x_j; \kappa_j,\gamma_j)=\frac{\kappa_j\gamma_j ^{\kappa_j}}{(x_j+\gamma_j)^{\kappa_j +1}},\quad \kappa_j,\gamma_j >0\quad\text{on}\quad x_j>0.
\end{equation}
Suppose $\bX=(\sqc{X}{d})$ has a survival Clayton copula with the density given by
\begin{equation}\label{survival clayton copula}
c(\bu;\theta)=\frac{\theta^{d} \Gamma(\frac{1}{\theta}+d)}{\Gamma(\frac{1}{\theta})}
 \left\{
\prod_{j=1}^{d}(1-u_{j})^{-\theta-1}
 \right\}
 \left\{
 \sum_{j=1}^{d}(1-u_{j})^{-\theta}-d+1
 \right\}^{-\frac{1}{\theta}-d},\quad 0<\theta<\infty.
\end{equation}
Some simple calculations show that the marginal distribution \eqref{pareto density} satisfies $(C2)$ and the copula \eqref{survival clayton copula} satisfies $(C3)$ under a very mild sufficient condition that $0<\theta < \log(1-p)/\log(1-\frac{1}{d})$.
Therefore, with any choice of proposal distribution $q$ satisfying $(C1)$, the corresponding MH estimator \eqref{MCMC estimator} satisfies consistency and asymptotic normality.
A possible choice of $q$ is the random walk proposal $q(\bx,\by)=f(\by-\bx)$ with $f$ the density of multivariate normal distribution with mean zero. 
Since $\by-\bx \in [-v,v]^{d'}$ for $\bx,\by \in \mathcal S_v$, $f(\by-\bx)$ is always positive.
\end{example}

It is worth noting that the condition $(C3)$ is irrelevant to the copula on the upper tail part $[F_{1}(v),1]\times\cdots \times [F_{d}(v),1]$. 
Therefore, $(C3)$ holds even if a copula density explodes at the upper corner, which is often the case with copulas having upper tail dependence.
In fact, a more general result holds for survival Archimedean copulas.
A $d$-dimensional \emph{Archimedean copula} with an Archimedean generator $\psi$ is given by
\begin{align}\label{archimedean copulas}
C_{\psi}(\bu) = \psi\left(\sum_{j=1}^{d}\psi^{-1}(u_j)\right),
\end{align}
where $\psi$ is a continuous and nonincreasing function $\psi:[0,\infty]\rightarrow [0,1]$ satisfying $\psi(0) = 1$, and $\lim_{t \rightarrow \infty}\psi(t)=0$, and is decreasing on $[0,\inf\{t:\psi(t)=0 \}]$.
The inverse $\psi^{-1}(u)$ is well-defined on $u \in (0,1]$ and $\psi^{-1}(0)$ is defined by $\psi^{-1}(0)=\inf\{t:\psi(t)=0\}$. 
Let $\psi^{(j)}$ be the $j$th derivative of $\psi$. 
An Archimedean generator $\psi$ defines a proper $d$-copula via \eqref{archimedean copulas} for any $d\geq 1$ if and only if $\psi$ is \emph{completely monotone}, that is, $(-1)^{j}\psi^{(j)}\geq 0$ on $(0,\infty)$ for all $j = 0,1,\dots$; see \cite{mcneil2009multivariate}.
We denote the class of completely monotone generators as $\Psi_{\infty}$.
According to Bernstein's Theorem \citep[see, for example,][]{feller2008introduction}, $\psi \in \Psi_{\infty}$ admits the Laplace-Stieltjes representation $\psi(t)=\E_{F}[{\mathrm e}^{-tV}]$ for some positive random variable $V>0$.

\begin{theorem}[Sufficient condition of $(C3)$ for survival Archimedean copulas]\label{theorem of survival archimedean copulas}
Let $\psi \in \Psi_{\infty}$ be a completely monotone Archimedean generator. 
If $\E[V^d] < \infty$ where $V$ is such that $\psi(t)=\E[{\mathrm e}^{-tV}]$, then the survival Archimedean copula $\bar C_{\psi}$ has a density satisfying the condition $(C3)$ in Theorem~\ref{main theorem}; moreover, $\bar C_{\psi}$ has a zero lower tail dependence coefficient.
\end{theorem}
\begin{proof}

Denote $\bar u_j=F_j(v)<1$ and $\underbar{$u$}_j= 1-\bar u_j>0$. 
The density of the survival Archimedean copula is given by
\begin{align}\label{Archimedean density}
\nonumber \bar c_{\psi}(\bu) &= c_{\psi}(1-\bu)
= \psi^{(d)}\left(\sum_{j=1}^{d}\psi^{-1}(1-u_j)\right)\prod_{j=1}^{d}\frac{1}{\psi^{(1)}(\psi^{-1}(1-u_j))}\\
&= (-1)^d\psi^{(d)}\left(t\right)\prod_{j=1}^{d}\frac{1}{(-1)\psi^{(1)}(t_j)},
\end{align}
where $t_j = \psi^{-1}(1-u_j)$ and $t=\sum_{j=1}^{d}t_j$.
When $u_j \in [0,F_j(v)]$, we have $0<\underbar{$u$}_j \leq 1-u_j \leq 1$ and thus $t_j = \psi^{-1}(1-u_j) \in [0,\bar t_j]$ where $\bar t_j = \psi^{-1}(\underbar{$u$}_j) < \infty$.
Thus, $0 \leq t=\sum_{j=1}^{d}l_j < \infty$.

Since $\psi \in \Psi_{\infty}$, it is of the form $\psi(t)=\E[{\mathrm e}^{-tV}]$ for some positive random variable $V>0$.
Therefore, on $0 \leq t <\infty$, we have $0< (-1)^j \psi^{(j)}(t)<\infty$ for $j=1$ and $j=d$ since $(-1)^j \psi^{(j)}(t)=\E[V^j {\mathrm e}^{-tV}]>0$ and  $\E[V^j {\mathrm e}^{-tV}] \leq \E[V^j] < \infty$ for $j=1$ and $j=d$ by assumption.
Consequently, the density~\eqref{Archimedean density} is bounded from below and above.

When $\E[V^d]<\infty$, the corresponding Archimedean copula has an upper tail dependence coefficient
\begin{align*}
\lambda_{u}(C_{\psi})=2-2\lim_{t\rightarrow 0}\frac{1-\psi(2t)}{1-\psi(t)}
=2-2\lim_{t\rightarrow 0}\frac{\psi^{(1)}(2t)}{\psi^{(1)}(t)}
\end{align*}
where the last equality comes from l'H$\hat{\text{o}}$pital's rule. 
Since 
\begin{align*}
\lim_{t\rightarrow 0}\frac{\psi^{(1)}(2t)}{\psi^{(1)}(t)}=
\lim_{t\rightarrow 0}\frac{(-1)\psi^{(1)}(2t)}{(-1)\psi^{(1)}(t)}\\
=\lim_{t\rightarrow 0}\frac{\E[V\mathrm{e}^{-2tV}]}{\E[V\mathrm{e}^{-tV}]}=1
\end{align*}
since $\E[V\mathrm{e}^{-2tV}]$ and $\E[V\mathrm{e}^{-tV}]$ go to $\E[V]<\infty$ as $t \rightarrow 0$.
Thus, for the survival Archimedean copula, $\lambda_{l}(\bar C_{\psi})=\lambda_{u}(C_{\psi})=0$.
\end{proof}

According to the Theorem~\ref{theorem of survival archimedean copulas}, the survival Clayton copula satisfies $(C3)$ while the survival Gumbel copula does not because it has a positive lower tail dependence coefficient; see \citet{hofert2012likelihood}.

\begin{remark}[Consistency and CLT for copulas with lower tail dependence]
Condition $(C3)$ does not hold for elliptical copulas with lower tail dependence, such as a student's $t$-copula with density
\begin{equation}\label{student t copula density}
c^{t}(\bu;\nu,{\bf P})=\frac{\Gamma(\frac{\nu+d}{2})\Gamma(\frac{\nu}{2})}{|{\bf P}|^{\frac{1}{2}}\Gamma(\frac{\nu+1}{2})^{d}}
\frac{
  \left(
    1+\frac{\tra{\bx}{\bf P}^{-1}\bx}{\nu}
  \right)^{-\frac{\nu+d}{2}}
}{
  \Pi_{j=1}^{d}(1+\frac{x^{2}_{j}}{\nu})^{-\frac{\nu+1}{2}}
},\quad\nu>0.
\end{equation}
By carefully checking the proof of Theorem \ref{main theorem},  the consistency and asymptotic normality of $\hat{{\bm \pi}}_{N}(\bm h)$ still hold under a weaker condition than $(C2)$ and $(C3)$;
\begin{align}\label{sufficient condition for CLT on q and pi}
\frac{q(\by,\bx)}{\pi(\bx)} = 
 \frac{q(\by,\bx)f_{S}(v)}{c(F_1(x_1),\dots,F_{d}(x_d))f_1(x_1)\cdots f_{d}(x_d)} \geq L,\quad \bx,\by \in \tilde{\mathcal S}_v,
\end{align}
for some positive constant $L>0$, where $\tilde{\mathcal S}_v = \{\bx \in \R^{d}:\tra{\bone}_{d}\bx = 1\}$.
While it is not straightforward to determine, one sufficient condition of \eqref{sufficient condition for CLT on q and pi} under $(C1)$ is that $\pi$ is bounded above on $\tilde{\mathcal S_v}$.
Another condition is that the proposal density $q$ explodes faster than $\pi$.
An example of such $q$ can be an independent proposal distribution $q(\bx,\by)=f(\by)$ with $f$ the density of the Dirichlet distribution $\mathcal D(\sqc{\alpha}{d})$ for $\sqc{\alpha}{d}<1$, which explodes to $\infty$ as $\bx$ approaches to an axis.
Therefore, by choosing such proposal distributions, consistency and CLT can still hold even if a copula density explodes at the lower corner $\bu = \bzero$.
\end{remark}

\subsection{Case of profits and losses}

In contrast to the case of pure losses, showing the consistency and CLT of the MH estimator is challenging for the case wherein we model P\&L.
Since the conditional density $f_{\bX'|S=v}$ is supported on the unbounded space $\R^{d'}$, careful study of its tail behaviors is necessary.
When the original loss random variable $\bX$ follows an elliptical distribution, the results of \citet{kamatani2017ergodicity} can be applicable to justify the CLT of our MH estimator with the MpCN proposal distribution.
An example of justification of CLT for the case wherein $\bX$ follows the multivariate student's $t$-distribution is provided below.

\begin{example}[Justification of CLT for multivariate student's $t$-Distribution]
We demonstrate that the MpCN proposal distribution $(\ref{MpCN proposal})$ achieves the CLT of VaR contributions when the underlying loss model is a multivariate student's $t$-distribution $t_{\nu}(\bmu,{\bf \Sigma})$ with density
\begin{equation}\label{multivariate student t density}
f_{\bX}(\bx;\nu,\Sigma)=\frac{\Gamma(\frac{\nu+d}{2})}{|\pi d{\bf \Sigma}|^{\frac{1}{2}}\Gamma(\frac{\nu}{2})} 
\left(
1+\frac{
  \tra{(\bx-\bmu)}{\bf \Sigma}^{-1} (\bx-\bmu)
}{\nu}
\right)^{-\frac{\nu+d}{2}}.
\end{equation}

Let $\bX \sim t_{\nu}(\bmu,{\bf \Sigma})$ where $\nu>2$, $\bmu \in \R^{d}$, and ${\bf \Sigma} \in {\mathcal M}^{d\times d}_{+}$.
Throughout the discussion, we set $\bmu=\bzero$ for simplicity.
Write  
\begin{equation*}
{\bf \Sigma}^{-1} = \left(
    \begin{array}{cc}
      {\bf A}_{1} & {\bm a}_{2} \\
      \tra{{\bm a}_{2}} & a_{3}  
    \end{array}
  \right)=:{\bf A}
\end{equation*}
for ${\bf A}_{1}\in {\mathcal M}^{d'\times d'}(\R)$, ${\bm a}_{2} \in \R^{d'}$, and $a_{3}\in \R$.
Then, it holds that
\begin{equation*}\label{conditional quadratic form}
\left(
    \tra{\begin{array}{c}
      \bx  \\
      v-\tra{\bone_{d'}}\bx  
    \end{array}
  \right)}
  \left(
    \begin{array}{cc}
      {\bf A}_{1} & {\bm a}_{2} \\
      \tra{{\bm a}_{2}} & a_{3}  
    \end{array}
  \right)
 \left(
    \begin{array}{c}
      \bx  \\
      v-\tra{\bone_{d'}}\bx  
    \end{array}
  \right)
  =
  \tra{(\bx-\bw)}{\bf V}(\bx-\bw)+\eta,
\end{equation*}
where ${\bf V}:={\bf A}_{1}-{\bm a}_{2}\tra{\bone_{d'}}-\bone_{d'}\tra{{\bm a}_{2}}+\bone_{d'}\tra{\bone_{d'}} \in {\mathcal M}^{d'\times d'}_{+}$, $\bw:={\bf V}^{-1}(vA_{3}\bone_{d'}-v{\bm a}_{2}) \in \R^{d'}$, and $\eta:=v^{2}a_{3}-\tra{\bw}{\bf V}\bw \in \R$.
Using this identity, we have that
\begin{eqnarray}\label{pearson type vii}
f_{\bX'|S=v}(\bx)&\propto&f_{\bX}(\bx,v-\tra{\bone_{d'}}\bx) \nonumber \\
&\propto&
\left(
1+\frac{
  \tra{(\bx-\bw)}{\bf W}^{-1} (\bx-\bw)+\eta
}{\nu}
\right)^{-\frac{\nu+d}{2}} \nonumber \\
&\propto&
\left(
1+\frac{
  \tra{(\bx-\bw)}{\bf W}^{-1} (\bx-\bw)
}{\nu+\eta}
\right)^{-\frac{\nu+d}{2}},
\end{eqnarray}
where $\bf W = \bf V ^{-1}$.
Provided $\nu+\eta>0$, $\bX'|S=v$ follows a $d'$-dimensional elliptical distribution with the location parameter $\bw$, scale parameter ${\bf W}$, and the density generator $g: \R_{+} \rightarrow \R_{+}$ given by
\begin{equation*}\label{density generator}
g(x)=\left(
1+\frac{
  x
}{\nu+\eta}
\right)^{-\frac{\nu+d}{2}}.
\end{equation*}
This type of distribution is called a {\it Pearson type $V\hspace{-1mm}I\hspace{-1mm}I$ distribution} \citep{schmidt2002tail}.

Consider the MH estimator $(\ref{MCMC estimator})$ where target distribution $\pi$ is $f_{\bX'|S=v}$, and proposal distribution $q$ is MpCN $(\ref{MpCN proposal})$.
According to Theorem 25 in \citet{roberts2004general}, $\sqrt{N}$-CLT holds if the Markov chain is geometrically ergodic and $\E[||\bX'||^{2}|S=v]<\infty$.
According to Proposition 3.4 in Kamatani (2016), the Markov chain with the MpCN proposal distribution is geometrically ergodic if $\E[||\bX'||^{\delta}|S=v]<\infty$ for some $\delta>0$, $\pi(\bx)$ is strictly positive and continuous, and it is symmetrically regularly varying, that is,
\begin{equation}\label{symmetrically regularly varying}
\lim_{r \rightarrow \infty} \frac{\pi(r\bx)}{\pi(r\bone_{d'})}=\lambda(\bx),
\end{equation}
for some function $\lambda:\R^{d'}\rightarrow (0,\infty)$ such that $\lambda(\bx)=1$ for any $\bx \in S^{d'-1}_{{\bf W}}$, where $S^{d'-1}_{{\bf W}}:=\{\bx \in \R^{d'}: ||{\bf W}^{-\frac{1}{2}}\bx||=||{\bf W}^{-\frac{1}{2}}\bone_{d'}||\}$.
We will see that the moment condition holds, and the condition on tail $(\ref{symmetrically regularly varying})$ is also satisfied for $\pi=f_{\bX'|S=v}$.

Write $R:=||\bX'||$.
It can be shown that $g$ is regularly varying \citep[see, for example,][]{resnick2013extreme} at $\infty$ with index $\alpha=-\frac{\nu+d}{2}$; that is,
\begin{equation}\label{density generator regularly varying}
\lim_{r \rightarrow \infty}\frac{g(rx)}{g(r)}=x^{-\frac{\nu+d}{2}}, \quad x>0.
\end{equation}
According to Proposition 3.7 in \citet{schmidt2002tail}, $f_{R|S=v}$ is regularly varying with index $-(\nu+1)$.
Then, according to Karamata's Theorem \citep[we referred to][]{resnick2013extreme}, $F_{R|S=v}$ is regularly varying with index $-\nu$.
Therefore, $\E[R^{\delta}|S=v]<\infty$ holds for any $\delta<\nu$; see \citet
[][]{mikosch1999regular}.
Thus, all the moment conditions above are satisfied as long as $\nu>2$.
In the elliptical case, tail condition $(\ref{symmetrically regularly varying})$ is a direct consequence of $(\ref{density generator regularly varying})$.
Since $\tra{(\bx-\bw)}{\bf W}^{-1} (\bx-\bw)>0$ for all $\bx \in \R^{d'}$, it holds that
\begin{equation*}
\lim_{r \rightarrow \infty} \frac{f_{\bX'|S=v}(r\bx)}{f_{\bX'|S=v}(r\bone_{d})}=\left(
\frac{
||{\bf W}^{-\frac{1}{2}}\bx||
}{
||{\bf W}^{-\frac{1}{2}}\bone_{d'}||
}
\right)^{-(\nu+d)},\quad \bx \in \R^{d'}.
\end{equation*}
Thus, by taking 
\begin{equation*}
\lambda(\bx):=\left(
\frac{
||{\bf W}^{-\frac{1}{2}}\bx||
}{
||{\bf W}^{-\frac{1}{2}}\bone_{d'}||
}
\right)^{-(\nu+d)},
\end{equation*}
in $(\ref{symmetrically regularly varying})$, $\pi=f_{\bX'|S=v}$ is shown to be symmetrically regularly varying.
Putting them together, we conclude that the MH estimator with the MpCN proposal distribution satisfies $\sqrt{N}$-CLT when the underlying loss vector follows a multivariate student's $t$-distribution with $\nu>2$ and $\eta>-\nu$.
Note that in the numerical experiment in section $\ref{Simulation and empirical studies}$, we set $d=3$ and $\nu=4$. 
Since $\eta+\nu=137.935>0$, CLT holds true.

\end{example}

\end{document}